\RequirePackage{amsmath}
\documentclass[runningheads]{llncs}

\newtheorem{observation}{Observation}

\usepackage{graphicx}                    
\usepackage{todonotes}
\usepackage{float}
\usepackage[section]{placeins}

\usepackage{graphicx}
\begin{document}
\title{Online Domination: The Value of Getting to Know All your Neighbors}
%
%

\author{Hovhannes Harutyunyan \and
Denis Pankratov \and
Jesse Racicot}
\authorrunning{H. Harutyunyan et al.}
%
\institute{Department of Computer Science and Software Engineering, Concordia University, Montreal, QC, H3G1M8, Canada}

\maketitle              
\begin{abstract}
We study the dominating set problem in an online setting. An algorithm is required to guarantee competitiveness against an   adversary that reveals the input graph one node at a time. When a node is revealed, the algorithm learns about the entire neighborhood  of the node (including those nodes that have not yet been revealed).  Furthermore, the adversary is required to keep the revealed portion  of the graph connected at all times. We present an algorithm that  achieves 2-competitiveness on trees and prove that this competitive  ratio cannot be improved by any other  algorithm. We also present algorithms that achieve 2.5-competitiveness on cactus graphs, $(t-1)$-competitiveness on $K_{1,t}$-free graphs, and  $\Theta(\sqrt{\Delta})$ for maximum degree $\Delta$ graphs. We show that all of those competitive ratios are tight. Then, we study several more general classes of graphs, such as threshold, bipartite  planar, and series-parallel graphs, and show that they do not admit competitive algorithms (that is, when competitive ratio is  independent of the input size). Previously, the dominating set problem was considered in a slightly  different input model, where a vertex is revealed  alongside its  restricted neighborhood: those neighbors that are among already  revealed vertices. Thus, conceptually, our results quantify the  value of knowing the entire neighborhood at the time a vertex is  revealed as compared to the restricted neighborhood. For instance, it was known in the restricted neighborhood model that 3-competitiveness is optimal for trees, whereas knowing the neighbors allows us to improve it to 2-competitiveness.

\end{abstract}
\section{Introduction}
\label{sec:intro}
Given an undirected simple graph $G=(V,E)$, a subset of vertices $D \subseteq V$ is called \emph{dominating} if every vertex of $V$ is either in $D$ or is adjacent to some vertex in $D$. In the well-known $\mathcal{NP}$-hard dominating set problem, the goal is to find a dominating set of minimum cardinality. We study this problem in the online setting, where a graph is revealed one node at a time. When a node is revealed its entire neighborhood is revealed as well. An algorithm is required to make an irrevocable decision on whether to include the newly revealed vertex into the dominating set the algorithm is constructing or not. This decision must be made before the next vertex is revealed. Performance of an online algorithm is measured against an optimal offline algorithm, i.e., an algorithm that knows the entire input in advance and has infinite computational resources. This measure is captured by the notion of competitive ratio and analysis, which is made precise below. For now, it suffices to note that competitive ratio is analogous to approximation ratio in the offline setting.

The dominating set problem has important practical and theoretical applications, such as establishing surveillance service (\cite{Berge62}), routing and transmission services in (wireless) networks (\cite{DasB97}), as well as broadcasting (\cite{VertexAddition,BroadcastingDomination}). While the dominating set problem and its variants (connected dominating set, independent dominating set, weighted dominating set, etc.) have been extensively studied in the offline setting~\cite{Berge62,HaynesHS98,HenningA13,Konig50,Ore62,WangDC16}, this problem has received little attention in the online algorithms community. The current paper attempts to fill in this gap, while making a quantitative comparison with another online model for dominating set.

Online dominating set problem has been studied in the vertex arrival model by Boyar et al.~\cite{OnlineDominatingSet3}. In that model, when a vertex is revealed only restricted neighborhood of that vertex is revealed as well, namely, those neighbors that appear among previously revealed vertices. Moreover, in the model considered by Boyar et al. decisions are only partially irrevocable, i.e., when a vertex arrives an algorithm may add this vertex together with \emph{any of its neighbors from the restricted neighborhood} to the dominating set. Thus, the decision to include a vertex is irrevocable, while the decision not to include a vertex is only partially irrevocable -- an algorithm has a chance to reconsider when any yet unrevealed neighbors arrive. The catch is that the algorithm does not know the input size and has to maintain a dominating set at all times. In the model considered in this paper, all decisions (to include or exclude a vertex from a dominating set) are irrevocable. On one hand, this makes our model stronger for the adversary. On another hand, our model is weaker for the adversary than the model of Boyar et al. in the aspect of the adversary being forced to reveal all neighbors of a newly revealed vertex at once. Thus, our results when compared to those of the vertex arrival model can be viewed as quantifying the value of getting to know all neighbors of a vertex at the time of its revelation. 

Perhaps somewhat surprisingly, we discover in several results that the benefit of knowing all neighbors outweighs the drawbacks of fully irrevocable decisions. Our results are summarized below, but in particular we show that in our model $\Delta$-bounded degree graphs admit $O(\sqrt{\Delta})$ online algorithms, while Boyar et al. show that $\Omega(\Delta)$ is necessary in their model. Similarly, we demonstrate and analyze a $2$-competitive algorithm for trees, while Kobayashi~\cite{OnlineDominatingSet4} shows a lower bound of $3$ in the vertex arrival model. Our degree upper bound implies that $O(\sqrt{n})$ competitive ratio is tight for general graphs, whereas Boyar et al. showed the lower bound of $\Omega(n)$ in the vertex arrival model. This paints a picture that knowing all the neighbors improves not only precise constants, when graph classes allow for small competitive ratio algorithms, but also give asymptotic improvements for more ``challenging'' graph classes for algorithms.

Prior to summarizing our results, we give a brief overview of competitive analysis framework. For more details, an interested reader should consult excellent books~\cite{BorodinE1998,Komm16} and references therein. Let $ALG$ be an algorithm for the online dominating set problem. Let $ALG(G, \sigma)$ denote the set of vertices that are selected by $ALG$  on the input graph $G$ with its vertices revealed according to the order $\sigma$. We sometimes abuse the notation and omit $G$ or $\sigma$ (or both) when they are clear from the context. Abusing notation even more, we sometimes write $ALG(G, \sigma)$ to mean $|ALG(G, \sigma)|$. Similar conventions apply to an offline optimal solution denoted by $OPT$. We say that $ALG$ has \emph{strict competitive ratio} $c$ if $ALG \le c \cdot OPT$ on all inputs. We say that $ALG$ has \emph{asymptotic competitive ratio} $c$ (or, alternatively, that $ALG$ is $c$-competitive) if $\limsup_{OPT \rightarrow \infty} \frac{ALG}{OPT} \le c$. The \emph{competitive ratio} of $ALG$ is the infimum over all $c$ such that $ALG$ is $c$-competitive. When we simply write ``competitive ratio'' we typically mean ``asymptotic competitive ratio'' unless stated otherwise.

We shall consider performance of algorithms with respect to restricted inputs, specified by  various graph classes, such as trees, cactus graphs, series-parallel, etc. The above definitions of competitive ratios can be modified by restricting them to inputs coming from certain graph classes. We denote the competitive ratio of an algorithm $ALG$ with respect to the restricted graph class CLASS by $\rho(ALG,$CLASS$)$.

The following is a summary of our contributions with the section numbers where the results appear. 
\begin{itemize}
    \item tight competitive ratio $2$ on trees (Section~\ref{ssec:trees});
    \item tight competitive ratio $\frac{5}{2}$ on cactus graphs (Section~\ref{ssec:cactus});
    \item tight competitive ratio $\Theta(\sqrt{\Delta})$ on maximum degree $\Delta$ graphs (Section~\ref{ssec:degree});
    \item tight competitive ratio $t-1$ on $K_{1,t}$-free graphs (Section~\ref{ssec:claws});
    \item tight competitive ratio $\Theta(\sqrt{n})$ for threshold graphs (Section~\ref{ssec:threshold}), planar bipartite graphs (Section~\ref{ssec:planar}),  and series-parallel graphs (Section~\ref{ssec:series-parallel}).
\end{itemize}

We note that all our upper bounds are in terms of strict competitive ratios, and all our lower bounds, with the exception of $K_{1, t}$-free graphs, are in terms of asymptotic competitive ratios.\footnote{With the small caveat that the performance ratio for threshold graphs is measured as a function of input size for reasons provided later.}

\section{Preliminaries}
\label{sec:prelim}

In this section we describe definitions and establish notations that will be used frequently in the rest of the paper.
Let $G = (V, E)$ be a connected undirected graph on $n = |V| \geq 1$ vertices. For a subset of vertices $S \subseteq V$ we define the \emph{closed neighborhood} of $S$, denoted by $N[S]$, to be $S \cup \{ v \in V \mid \exists u \in S, \{ u, v \} \in E\}$. We use $\langle S \rangle$ to denote the subgraph of $G$ induced on $S$.

The vertices $V$ are revealed online in order $(v_1, ..., v_n)$. Since we consider the online input model where vertices are revealed alongside their neighbors, we distinguish between two notions: those vertices that are revealed by a certain time and those that are visible. More precisely, we have the following:
\begin{definition}
\begin{itemize}
\item  $v_i$ is \emph{revealed by time} $j$ if $i \leq j$. 
\item $v_j$ is \emph{visible at time} $i$ if it is either revealed by time $i$ or it is adjacent to some vertex revealed by time $i$.
\item $R_i$ denotes the set of all vertices revealed by time $i$.
\item $V_i$ denote the vertices visible at time $i$ (i.e. $V_i = N[R_i]$).
\end{itemize}
\end{definition}  

The adversary chooses the graph $G$ as well as the revelation order of vertices; however, the adversary is restricted to those revelation orders that guarantee that $\langle R_i \rangle$ is connected for all $i$. Thus, we observe that the process of revelation of a graph by the adversary is a natural generalization of the breadth-first search (BFS) and depth-first search (DFS) explorations of the graph. Thus, we can define the \emph{revelation tree} analogous to BFS and DFS trees. We need the following observation first:

\begin{observation}\label{obs : Parent}
If $v_j \in V_{i} \setminus V_{i - 1}$ with $i \geq 2$ then $v_{i}$ is the unique neighbour of $v_j$ in $\langle V_{i} \rangle$.
\end{observation}

In the preceding observation, we say that $v_j$ is a \emph{child} of $v_{i}$ and that $v_i$ is the \emph{parent} of $v_j$.  The edge $\{ v_i, v_j \}$ is called a \emph{tree edge}. The subgraph induced on the tree edges is the revelation tree. Any edge $\{ u, v \}$ where $u$ is not the parent of $v$ nor $v$ the parent of $u$ is called a \emph{cross edge}.

After the vertex $v_i$ is revealed together with its closed neighborhood $N[v_i]$, an online algorithm $ALG$ must make a decision $d_i \in \{0,1\}$, which indicates whether the algorithm takes this vertex to be in the dominating set or not. 

\begin{definition}Given an online algorithm $ALG$ we define:
\begin{itemize}
    \item $S_i$ denote the set of revealed vertices selected by $ALG$ after decision $d_i$ (i.e. $S_i = \{ v_k \mid d_k = 1, 1 \leq k \leq i \}$) where $S_0 = \emptyset$. 
    \item $D_i = N[S_i]$ is the set of vertices that are dominated after decision $i$.
    \item $U_i = V_i \setminus D_{i - 1}$ is the set of visible vertices undominated immediately before decision $d_i$. $U_0 = \emptyset$.
\end{itemize}
\end{definition}

A series of figures are provided below which illustrate the preceding definitions.  For these figures, and all others in this paper, the convention is that vertices that are shaded in gray are those selected by $ALG$, vertices with thicker boundaries belong to $OPT$, an edge that is dashed is a cross edge, and all the solid edges are tree edges.

\begin{figure}[h]
\centering
\includegraphics[scale=0.5]{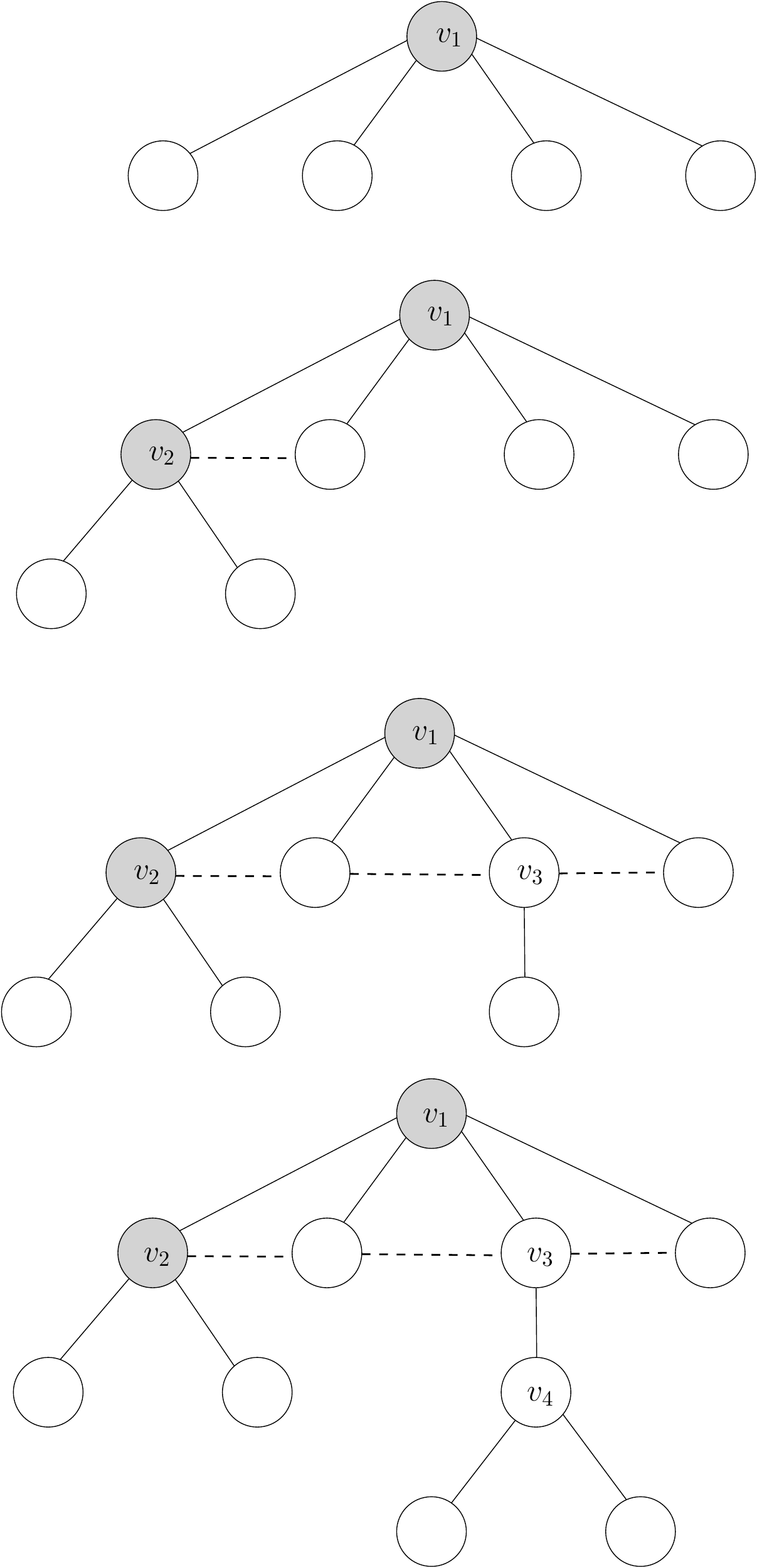}
\caption{A series of prefixes of the graph below.}\label{fig:basic-graph-full}
\end{figure}

\begin{figure}[h]
\centering
\includegraphics[scale=0.5]{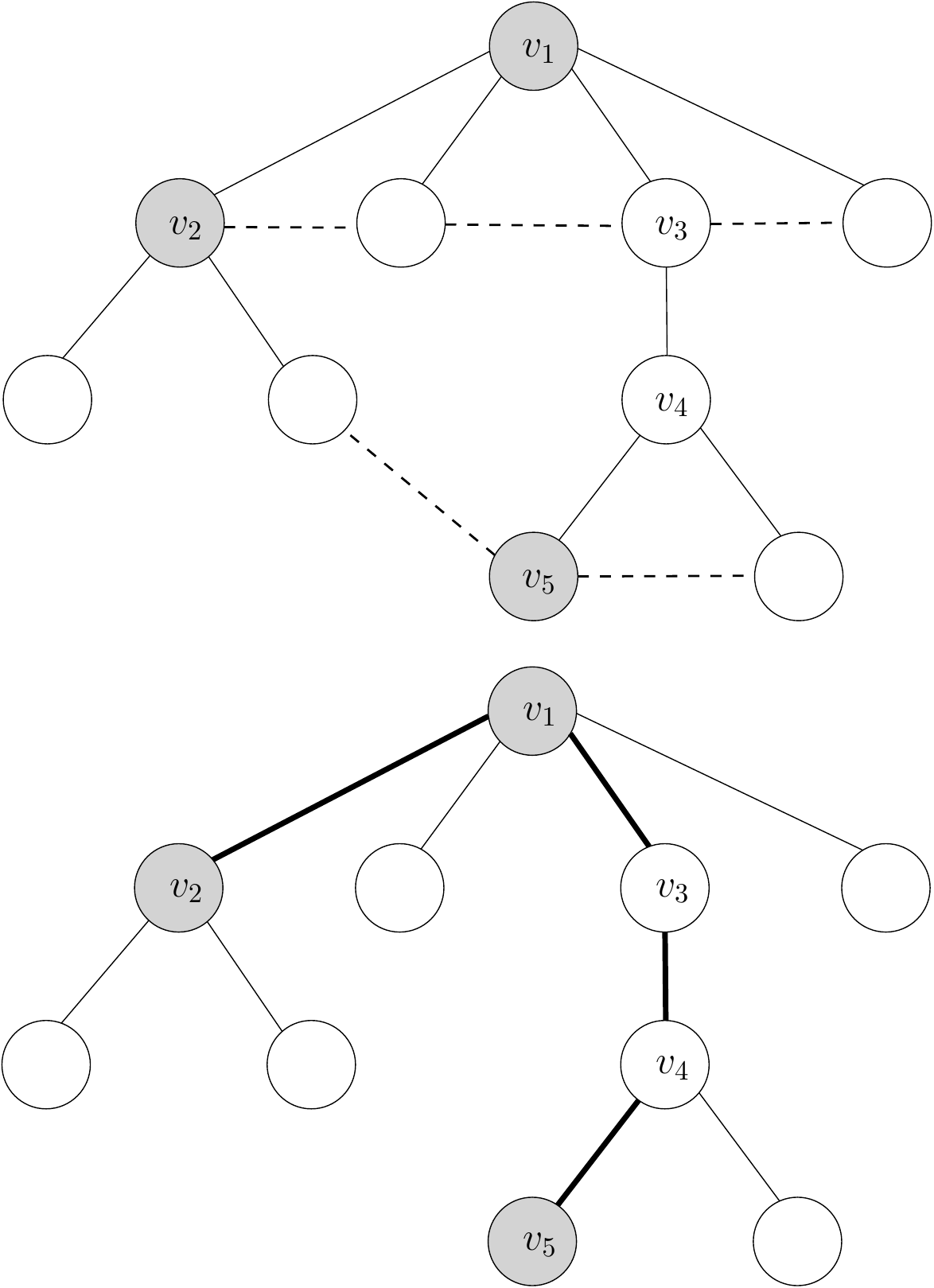}
\caption{Top picture depicts the prefix of a graph with $5$ revealed vertices and $10$ visible vertices.  Bottom picture depicts its revelation tree.  The thickened edges are the edges in the induced subgraph on the revealed vertices, which is acyclic here but is only guaranteed to be connected in general. }\label{fig:basic-graph-full}
\end{figure}

Since an online algorithm makes irrevocable decisions and it must produce a feasible solution, there may be situations where an algorithm is forced to select a vertex $v_j$ to be in the dominating set. This happens because $v_j$ is the ``last chance'' to dominate some other vertex $v_i$. In this case, we say that $v_j$ saves $v_i$ or that $v_j$ is the savior of $v_i$. Note that it is possible for a vertex $v_j$ to save itself. The following definition makes the notion of ``saving'' precise.

\begin{definition}
A vertex $v_j, j \geq 1$ \emph{saves} a vertex $v_i$ if $j = \max \{ k \mid v_k \in N[v_i] \}$ and $N[v_i] \setminus \{ v_j \}$ contains no vertices from $S_{j - 1}$.  Let $s(v_j)$ denote the set of vertices that $v_j$ saves.
\end{definition}

Observe that if a vertex is saved then it must be that every one of its neighbours (itself included) had a chance to dominate the said vertex.  

\begin{observation}\label{lemma: Saviour Lemma}
If $v_i$ is saved then $v_i \in N[v_j] \cap U_j$ for any $v_j \in N[v_i]$. 
\end{observation}

All our upper bounds are established by either a GREEDY algorithm or a $k$-DOMINATE algorithm for some fixed integer value of parameter $k$:
\begin{itemize}
\item The algorithm GREEDY selects a newly revealed vertex if and only if the vertex is not currently dominated. Using the notation introduced above, GREEDY selects $v_i, i \geq 1$ if and only if $v_i \in U_i$.  

\item The algorithm $k$-DOMINATE (for some fixed integer parameter $k$) selects a newly revealed vertex if and only if either (1) the vertex has at least $k$ undominated neighbors, or (2) the vertex saves at least one other vertex. Using the notation introduced above, $v_i, i \geq 1$ is selected if and only if either (1) $|N(v_i) \cap U_i| \geq k$, or (2) $|s(v_i)| \geq 1$.
\end{itemize}
Both GREEDY and $k$-DOMINATE give rise to  rather efficient offline algorithms so that any of the positive results given in this paper may be realized as efficient offline approximation algorithms.

\section{Competitive Graph Classes}
\label{sec:competitive}

\subsection{Trees}
\label{ssec:trees}

In this section we establish the tight bound of $2$ on the best competitive ratio when the input graph is restricted to be a tree. The upper bound is achieved by the $2$-DOMINATE algorithm and is proved in Theorem~\ref{thm : UB Trees} below. The lower bound on all online algorithms is established in Theorem~\ref{thm : LB Trees}. We begin this section with the lower bound.

\begin{theorem}\label{thm : LB Trees}
$\rho(ALG, TREE) \geq 2$ for any algorithm $ALG$.
\end{theorem}
\begin{proof}
Consider an arbitrary small $\epsilon > 0$. We will give an adversarial input that guarantees that $ALG \ge (2-\epsilon) OPT$. 
Let $k = \lceil \frac{3}{\epsilon} \rceil \geq 4$. At the start, the adversary reveals $v_1$ with $k$ children $\{c_1, \ldots, c_k\}$. Then we start the process described in the next paragraph at $c_1$. The process can terminate in two ways: (i) $ALG$ stops selecting vertices to be in the dominating set, or (ii) $ALG$ selects $k$ vertices revealed after $c_1$ (inclusive). If the process terminates because of (i), then the adversary restarts the process at child $c_2$ of $v_1$. The process again terminates either with (i) or (ii) with respect to $c_2$. If it is due to (i), then the adversary restarts the process at $c_3$, and so on. If the process terminates with (ii) with respect to $c_i$ then we reveal $c_j$ for $j > i$ as leaves of $v_1$. 

Next, we describe the process with respect to $c_i$. The adversary reveals $c_i$ with $2$ children and if $ALG$ selects $c_i$ then exactly one child of $c_i$ is revealed with two additional children. If $ALG$ selects the child then one of its children is revealed with two additional children, and so on. Let $j_i$ be the number of these vertices that are selected by $ALG$. This process terminates only if $ALG$ stops selecting these vertices with two children ($j_i < k$) or when $ALG$ selects $k$ of them ($j_i = k$).  At this point the subtree grown at $c_i$ has some revealed vertices as well as visible, but not yet revealed vertices. To finish revealing the entire subtree, the adversary proceeds as follows.

If $j_i < k$ then the two children on the $(j_i + 1)$'st vertex are revealed to be leaves.  Moreover, each of the $j_i$ selected vertices have exactly one visible child that is not yet revealed. Reveal those $j_i$ children, called \emph{support vertices}, with an additional leaf child (i.e. the child is revealed to be a leaf after its parent is revealed).  Including the $2$ children of the $(j_i + 1)$'st vertex $ALG$ must select at least $j_i + 2$ additional vertices to dominate these leaves for a total of $j_i + (j_i + 2) = 2(j_i + 1)$ selected vertices in this subtree. In this case, $OPT$ can select the support vertices together with the $(j_i+1)$'st vertex for a total $j_i+1$ vertices to dominate the entire subtree.

If $j_i = k$ the procedure to finish revealing the entire subtree at $c_i$ is similar: the $k$'th vertex children are both revealed to be leaves and each of the other $k - 1$ selected vertices has the other child become a support vertex, i.e., revealed with an additional leaf child.  The performance is similar here but $ALG$ is not forced to select the two children of the $k$'th vertex so $ALG$ selects at least $k + (k - 1) = 2k - 1$. In this case, $OPT$ needs only select the $k$'th vertex together with the support vertices for a total of $k$ vertices to dominate the subtree.  

To finish the analysis, we consider the following two cases: 

$\textbf{Case 1 :}$ for all $i$ we have $j_i < k$. Then $ALG \geq 2(j_i + 1)$ on each subtree whereas $OPT \leq j_i + 1$ on each subtree.  Summing over all subtrees and remarking that $OPT$ might select $v_1$ we obtain that

\[ ALG/OPT \ge \left( \sum 2(j_i+1) \right) / \left( 1+ \sum (j_i+1) \right) \ge 2 - 2/k \ge 2-\epsilon. \]

$\textbf{Case 2 :}$ there exists $\ell$ such that $j_\ell = k$. Then $OPT$ selects $j_i+1$ vertices for $i < \ell$, $k$ vertices for $i = \ell$, $0$ vertices for $i > \ell$ per subtree, plus $v_1$. Whereas $ALG$ selects at least $2(j_i+1)$ for $i < \ell$, $2k-1$ for $i = \ell$, and $0$ for $i > \ell$. By a similar calculation to \textbf{Case 1}, we obtain that $ALG/OPT \ge 2 - 3/k \ge 2 - \epsilon$.
\end{proof} 

\begin{figure}[h]
\centering
\includegraphics[scale=0.6]{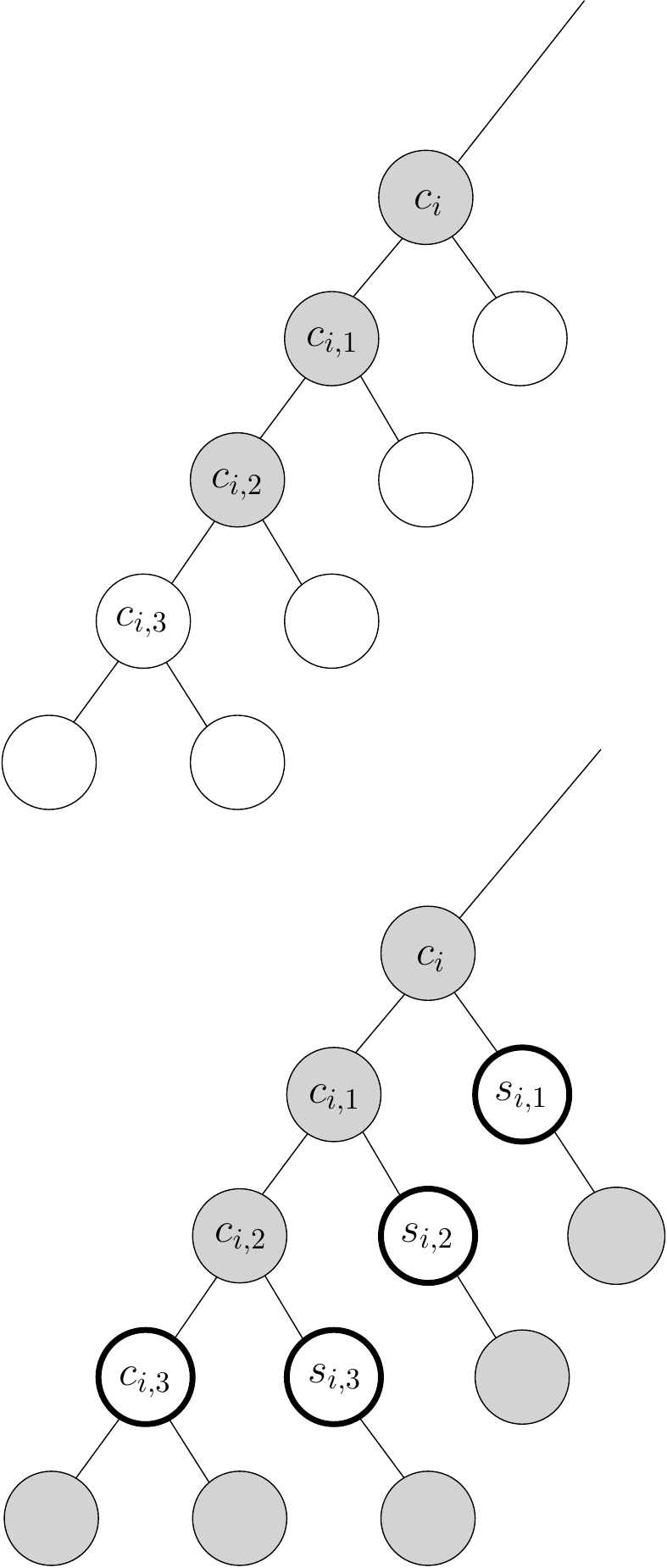}
\caption{An example of the process described in Theorem~\ref{thm : LB Trees} where $ALG$ selects $j_i = 3$ vertices on the subtree rooted at $c_i$.  The top depicts the subtree immediately after revealing $c_{i, 3}$ whereas the bottom shows the entirely revealed subtree.}\label{fig:cactus-2-gadget}
\end{figure}

Now that we have established an asymptotic lower bound of $2$ for any algorithm we show that $2$-DOMINATE is $2$-competitive.

\begin{theorem}\label{thm : UB Trees}
$\rho($2-DOMINATE, TREE$) = 2$.
\end{theorem}

\textit{High level overview of the proof.} Consider an arbitrary input $T = (V, E)$ on $n \geq 3$ vertices and let $OPT$ denote a minimum dominating set of $T$ which contains no vertices of degree $1$ (i.e. any such vertex can be exchanged for its only neighbor). Recall that $S$ is the set of vertices selected by $2$-DOMINATE. Initially, we assign charge $1$ to each vertex $v$ in $S$ and charge $0$ to each vertex $v$ not in $S$. Thus, $|S| = \sum\limits_{v \in S} ch(v)$ where $ch(v)$ denotes the charge of $v$. With a charging scheme described shortly, we spread the charge from the vertices in $S$ to the vertices of $V$. Let $ch^*(v)$ denote the new charge associated with vertex $v$. We extend the functions $ch$ and $ch^*$ to subsets of vertices linearly, e.g., for $W \subseteq V$ we have $ch(W) = \sum_{v \in W} ch(v)$. We shall demonstrate that the procedure of spreading the charge satisfies two properties:
\begin{enumerate}
    \item conservation property: $\sum_v ch(v) = \sum_v ch^*(v)$ meaning that the total charge is preserved; and
    \item $OPT$-concentration property: for each $v \in OPT$ we have $ch^*(N[v]) \le 2$.
\end{enumerate}
With these two properties it follows that $2$-DOMINATE $ \le \sum_v ch(v) = \sum_v ch^* (v) \le \sum_{v \in OPT} ch^*(N[v]) \le 2 OPT$, so $2$-DOMINATE is strictly $2$-competitive.

Before we proceed with this plan, we make a couple of useful observations:
\begin{lemma}\label{Lemma : Tree Cross Edges}
If input is a tree, there are no cross edges incident on any vertex $v_i$.  In particular, any vertex $v_i$ has at most one neighbour before it is revealed.
\end{lemma}


\begin{corollary}\label{cor:  DegThree}
If $deg(v_i) \geq 3$ then $v_i \in S$.  
\end{corollary}

Now, we are ready to present formal details of the above plan. We spread the charges according to the following rule: 

Consider any $v_i \in S$ with $X_i = N[v_i] \cap U_i$.  Remarking that $X_i \neq \emptyset$ we then give each vertex in $X_i$ an equal charge of $\frac{1}{|X_i|}$.  That is, a vertex selected by $2$-DOMINATE spreads its charge evenly to all the newly dominated vertices in its closed neighbourhood. We say that each vertex in $X_i$ is charged by $v_i$.

\begin{observation}\label{obs: Charged Once}
Every vertex is charged by exactly one vertex.
\end{observation}

The preceding observation immediately implies that any vertex has charge at most $1$.  This observation is tight in the sense that, on certain inputs, there are vertices with charge equal to $1$.  A vertex with charge $1$ is a rather special case though.  In particular, if $v_i$ has charge $1$ then it must be saved by some vertex $v_{j}$ where $X_{j} = \{ v_i \}$ (this does not exclude the possibility that $v_i = v_{j}$).  If $v_i$ does not meet this condition then it must have charge at most $\frac{1}{2}$. 

\begin{lemma}\label{Tree Common Neighbour}
If $v_i$ and $v_j$ both have charge equal to $1$ then they share no common neighbours.
\end{lemma}
\begin{proof}
Suppose for the sake of deriving a contradiction that $v_{i'}$ were a common neighbour of $v_i$ and $v_j$.  Since $v_i$ is saved, by Observation $\ref{lemma: Saviour Lemma}$ it must be that $v_i \in N(v_i') \cap U_{i'}$.  Similarly, we have that $v_j \in N(v_i') \cap U_{i'}$.  That is, $|N(v_i') \cap U_{i'}| \geq 2$ and thus $v_i' \in S$.  Moreover, $X_{i'} = N[v_{i'}] \cap U_{i'}$ contains $v_i$ and $v_j$.  In particular, we have that $|X_{i'}| \geq 2$ with $v_i, v_j \in X_{i'}$ and therefore $v_i$ and $v_j$ receive charge no larger than $\frac{1}{2}$, a contradiction.
\end{proof}

\begin{lemma}\label{lemma: Tree Charge}
If $v_i$ and $v_j$ both have charge equal to $1$ then they are not adjacent.
\end{lemma}
\begin{proof}
It is easy to see that $v_1$ cannot have charge $1$ on any input with at least $2$ vertices.  Therefore we safely assume that $1 < i < j$ such that both $v_i$ and $v_j$ have a parent.  We assume for the sake of deriving a contradiction that $v_i$ and $v_j$ are adjacent.

Now, since both $v_i$ and $v_j$ have charge $1$ it follows that they are both saved vertices.  First we show that both $v_i, v_j \notin S$.  Notice that any saved vertex $v_k$ has the property that $|N[v_k] \cap S| = 1$.  Therefore, if we assume by way of contradiction that $v_i \in S$ we obtain that $N[v_i] \cap S = N[v_j] \cap S = \{ v_i \}$ and therefore $v_i$ saves itself and $v_j$.  This yields that $X_i = N[v_{i}] \cap U_{i}$ contains $v_i$ and $v_j$.  In particular, we have that $|X_{i}| \geq 2$ with $v_i, v_j \in X_{i}$ and therefore $v_i$ and $v_j$ receive charge no larger than $\frac{1}{2}$, a contradiction.  An identical argument will yield that $v_j \notin S$.

Therefore it must be that $v_i$ is saved by some vertex $v_{i'}$ with $i' \notin \{ i, j \}$.  Moreover, we must have $i < j < i' $ since $i < j$ by assumption and $i' = \max \{ k \mid v_k \in N[v_i] \}$.  This implies that both $v_j, v_{i'}$ are children of $v_i$ by Observation $\ref{Lemma : Tree Cross Edges}$ yielding that $|N(v_i) \cap U_i| \geq 2$ but $v_i$ cannot be in $S$. 
\end{proof}

From the two preceding lemmas we have the immediate corollary.

\begin{corollary}\label{Cor : Neighbourhood Charge Trees}
For any vertex $v_i$, at most one vertex in $N[v_i]$ has charge $1$.
\end{corollary}

Now, we finish the proof of $2$-competitiveness of $2$-DOMINATE on trees.

\begin{proof}[Proof of Theorem~\ref{thm : UB Trees}]
The lower bound follows from Theorem \ref{thm : LB Trees}.  Let $v_i \in OPT$ be an arbitrary vertex in $OPT$.  We consider two cases $\textbf{(1)}$ $deg(v_i) = 2$ or $\textbf{(2)}$ $deg(v_i) \geq 3$.

$\textbf{Case 1 :}$ Suppose that $deg(v_i) = 2$ and hence $|N[v_i]| = 3$.  By Corollary $\ref{Cor : Neighbourhood Charge Trees}$ it follows that at most one vertex in $N[v_i]$ has charge $1$.  If no vertices in $N[v_i]$ have charge $1$ then $ch(x) \leq \frac{1}{2}$ for each $x \in N[v_i]$ and we obtain that $\sum\limits_{x \in N[v_i]}ch(x) \leq 3\big(\frac{1}{2}\big) < 2$.  If there is exactly one vertex $x' \in N[v_i]$ with charge $1$ we therefore obtain that $\sum\limits_{x \in N[v_i]}ch(x) = \sum\limits_{x \in N[v_i] \setminus \{ x' \}}ch(x) + ch(x') \leq \frac{2}{2} + 1 = 2$.

$\textbf{Case 2 :}$ Suppose that $deg(v_i) \geq 3$.  By Corollary \ref{cor: DegThree} it follows that $v_i \in S$ with at least $2$ children.  Let $C_i = V_i \setminus V_{i - 1}$ denote the children of $v_i$ and remark that $C_i \subseteq X_i$.  That is, each child of $v_i$ is charged by $v_i$ and only $v_i$.  Therefore the children of $v_i$ can receive at most the full initial charge on $v_i$ and thus attribute a charge of at most $1$.  

Now we claim that any vertex in $N[v_i] \setminus C_i$ has a charge of at most $\frac{1}{2}$.  Indeed, suppose a vertex $v_{i'} \in N[v_i] \setminus C_i$ has charge $1$ then it must be saved by $v_i$ since $|N[v_{i'}] \cap S| = 1$ for any saved vertex $v_{i'}$.  That is, there is exactly one vertex in its closed neighbourhood that is selected and since $v_i$ is selected it must be $v_i$.  Thus, we must have that $v_{i'} \in X_i$ but since $C_i \subseteq X_i$ we know that $|X_i| \geq 2$ and thus $v_{i'}$ receives a charge of no more than $\frac{1}{2} < 1$, contradicting our assumption that $v_{i'}$ has charge $1$. 

Thus, by remarking that $|N[v_i] \setminus C_i| \leq 2$ we obtain that $\sum\limits_{x \in N[v_i]}ch(x) = \sum\limits_{v_{j} \in C_i}ch(v_j) + \sum\limits_{v_{i'} \in N[v_i] \setminus C_i}ch(v_{i'}) \leq 1 + 2\big(\frac{1}{2}\big) = 2$ as desired.
\end{proof} 

\subsection{Cactus Graphs}
\label{ssec:cactus}

A graph $G$ is said to be a cactus graph if it is connected and every edge lies on at most one cycle.  \cite{hedetniemi1986linear} provide an exact offline algorithm that runs in linear time for finding a minimum dominating set of a cactus graph.  Of course, an efficient offline algorithm does not guarantee that an online algorithm can perform well but fortunately, cactus graphs are a class of graphs for which an online algorithm can achieve constant competitive ratio.  In this section, we show that $2$-DOMINATE is $\frac{5}{2}$-competitive when inputs are restricted to cactus graphs, and that this is as well as any algorithm can perform.

Before presenting a lower bound of $\frac{5}{2}$ on all online algorithms we describe a gadget that is used in the proof.  The gadget itself is a cactus graph on $3 \leq n \leq 4$ vertices with the property that $OPT$ selects exactly $1$ vertex and any algorithm $ALG$ selects at least $2$ vertices.  Consider revealing a root vertex $r$ with $2$ children $c$ and $c'$.  If $ALG$ does not select $r$ then  both $c, c'$ are revealed as only adjacent to $r$ and $ALG$ must select both whereas $OPT$ selects only $r$.  If $ALG$ does select $r$ then $c$ is revealed as adjacent to $c'$, and $c'$ is revealed with an additional child $x$.  The vertex $x$ is adjacent only to $c'$ and thus $ALG$ must select at least one of $c', x$ whereas $OPT$ selects only $c'$ (both cases are depicted in figure \ref{fig:cactus-2-gadget}).  Given any input cactus graph with a visible vertex $r$ not yet revealed this gadget can be constructed with $r$ as the root.  Within the proof of the lower bound we call this a $2$-gadget. 

\begin{figure}[h]
\centering
\includegraphics[scale=0.5]{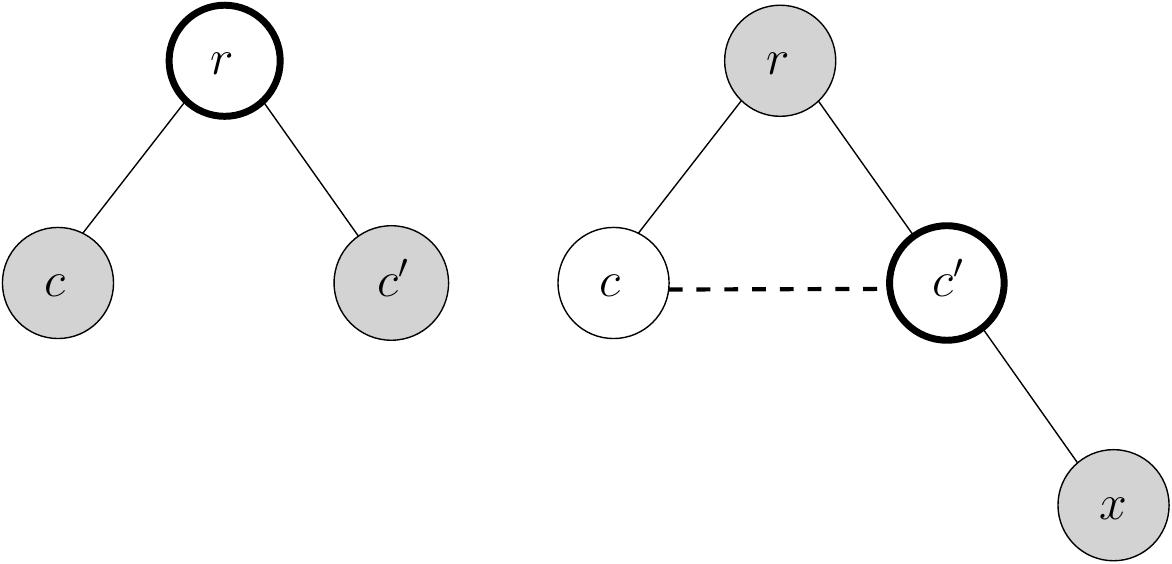}
\caption{The cactus $2$-gadget : The leftmost figure depicts the case where $ALG$ does not select the root $r$ and rightmost depicts the case where $ALG$ selects $r$.}\label{fig:cactus-2-gadget}
\end{figure}

\begin{theorem}\label{THM : LB Cactus}
$\rho(ALG, CACTUS) \geq \frac{5}{2}$ for any algorithm $ALG$.
\end{theorem}
\begin{proof}
We prove this lower bound on the asymptotic competitive ratio. We construct a sequence of graphs, one for each $k$, such that the competitive ratio approaches $5/2$ as $k$ goes to infinity. We start by revealing the first vertex $v_1$ with $k$ children. Then we run an adversarial process starting with each child of $v_1$ in order.  The process consists of several rounds, each round increases the output of $OPT$ and $ALG$. The process might terminate for one of two reasons: either (i) we guarantee strict competitive ratio at least $5/2$ on the subcactus rooted at the child, or (ii) we ran the process for sufficiently long time, i.e., $k$ rounds. Each round increases $OPT$ by a multiple of $2$ while increasing $ALG$ by a multiple of $5$. However, due to initial set up of the process $ALG$ might be off by additive $1$ from the intended multiple of $5$. When the process terminates according to (i), it means that $ALG$ made a mistake and this ``off by 1'' is corrected to give a strict competitive ratio. When the process terminates according to (ii), it means that the process ran for sufficient duration that the ``off by 1'' has been amortized and the ratio approaches $5/2$ asymptotically. After the first child of $v_1$  that terminates according to (ii) (if it exists), the rest of the children of $v_1$ are revealed as leaves. The formal analysis is analogous to that done in Theorem~\ref{thm : LB Trees} and is omitted. We present the process which constitutes the crux of the argument.

Using each child of $v_1$, we construct a subcactus for which $\frac{ALG}{OPT}$ approaches $\frac{5}{2}$. Let $c$ be a child of $v_1$ and reveal $c$ with $3$ children.  If $ALG$ does not select $c$ then each child of $c$ is revealed with no additional neighbours and $ALG$ must select all $3$ children whereas $OPT$ selects $c$.  Suppose then that $ALG$ selects $c$ and let $c_{1, 1}, c_{1, 2}, c_{1, 3}$ be the three children of $c$.  Reveal $c_{1, 1}$ as adjacent to $c_{1, 2}$ along with $2$ additional children.  If $ALG$ does not select $c_{1, 1}$ then the children of $c_{1, 1}$ are revealed as leaves, forcing $ALG$ to select them and $c_{1, 3}$ is revealed as the root of a $2$-gadget ($c_{1, 2}$ is revealed with no additional neighbours).  Thus, $\frac{ALG}{OPT} \geq \frac{5}{2}$ in this case (see Figure \ref{fig:cactus-second-rejected}).  If instead $ALG$ selects $c_{1, 1}$ then $c_{1, 2}$ and $c_{1, 3}$ are revealed as the roots of two distinct $2$-gadgets and since $c$ is dominated by $v_1$ (we assume that $v_1 \in OPT$) we have that $\frac{ALG}{OPT} \geq \frac{5}{2}$ on this subcactus (excluding $c_{1, 1}$) thus far and we continue the trap with $c_{1, 1}$ as the root.  

At this point, $c_{1, 1}$ is selected by $ALG$ we let $c_{2, 1}, c_{2, 2}$ be the $2$ children of $c_{1, 1}$ and we reveal $c_{2, 1}$ as adjacent to $c_{2, 2}$ with $2$ children $c_{3, 1}, c_{3, 2}$.  If $ALG$ does not select $c_{2, 1}$ then $c_{3, 1}, c_{3, 2}$ are revealed as leaves and $ALG$ selects $c_{1, 1}, c_{3, 1}, c_{3, 2}$ and $OPT$ can select $c_{2, 1}$ for a performance of $3$ along with the running performance of $\frac{5}{2}$ (see Figure \ref{fig:cactus-third-rejected}).  If $ALG$ does select $c_{2, 1}$ then $c_{3, 1}$ is revealed as adjacent to $c_{3, 2}$ with two children $c_{4, 1}, c_{4, 2}$.  If $ALG$ does not select $c_{3, 1}$ then $c_{4, 1}, c_{4, 2}$ are revealed as leaves and $c_{2, 2}$ is revealed with an additional leaf neighbour $l_{2, 2}$ so that $ALG$ must select at least one of $c_{2, 2}, l_{2, 2}$.  Thus, $ALG$ here selects $c_{1, 1}, c_{2, 1}, c_{4, 1}, c_{4, 2}$ and at least one of $c_{2, 2}, l_{2, 2}$ whereas $OPT$ can select $c_{3, 1}$ and $c_{2, 2}$ for a performance of $\frac{5}{2}$ (see Figure \ref{fig:cactus-fourth-rejected}).  If instead $ALG$ selects $c_{3, 1}$ (thus far $c_{1, 1}, c_{2, 1}$ and $c_{3, 1}$ are all selected) then $c_{2, 2}$ is revealed with an additional leaf neighbour $l_{2, 2}$ so that $ALG$ must select at least one of $c_{2, 2}, l_{2, 2}$, and $c_{3, 2}$ is revealed as the root of a $2$-gadget so that $\frac{ALG}{OPT} \geq \frac{5}{2}$ on the subcactus thus far (excluding $c_{3, 1}$) and we repeat the trap with $c_{3, 1}$ as the selected root (see Figure \ref{fig:cactus-fourth-accepted}).

\end{proof}

\begin{figure}[h]
\centering
\includegraphics[scale=0.6]{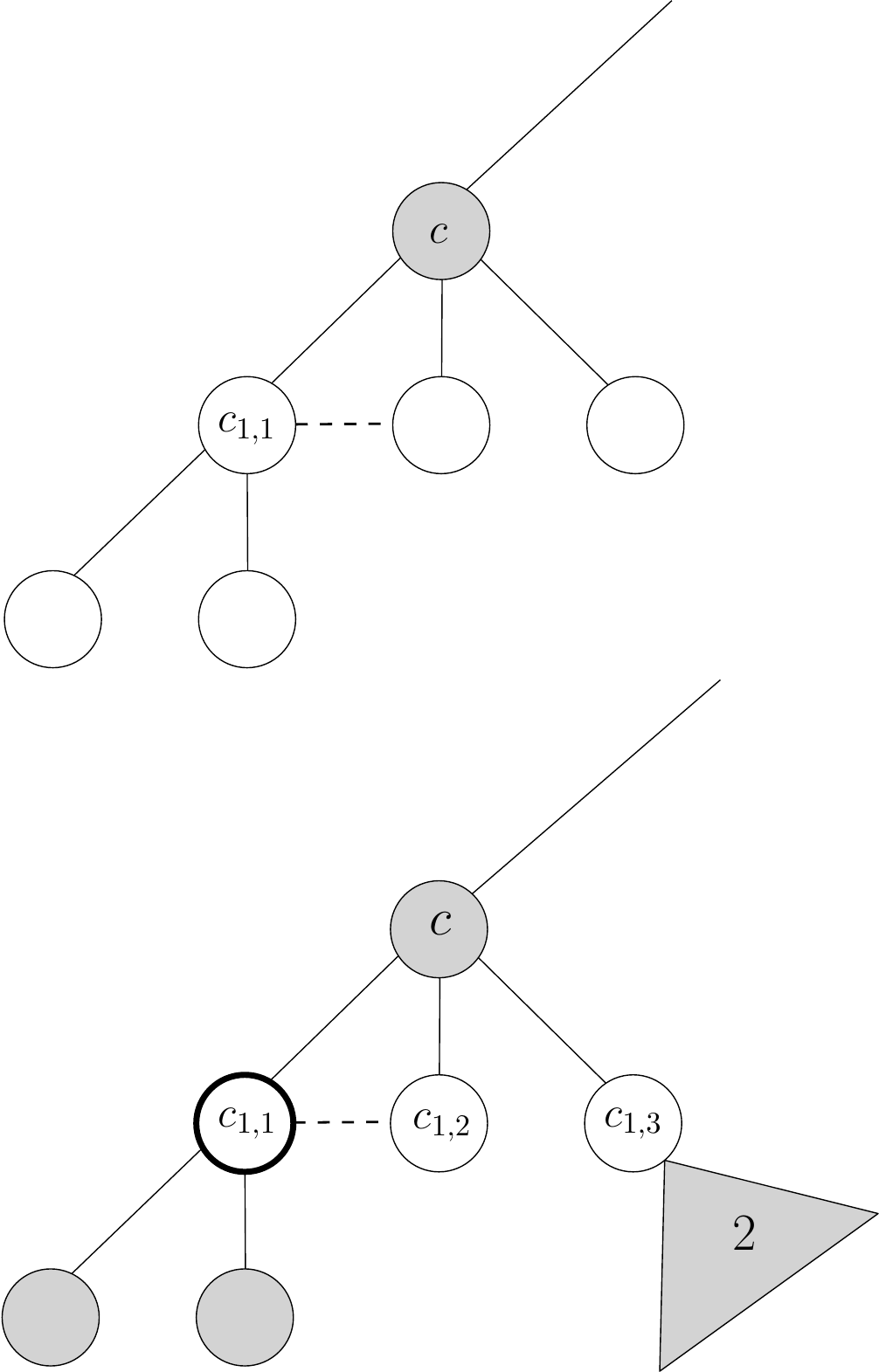}
\caption{The case described in Theorem~\ref{THM : LB Cactus} where $ALG$ does not select $c_{1, 1}$.  }\label{fig:cactus-second-rejected}
\end{figure}

\begin{figure}[h]
\centering
\includegraphics[scale=0.6]{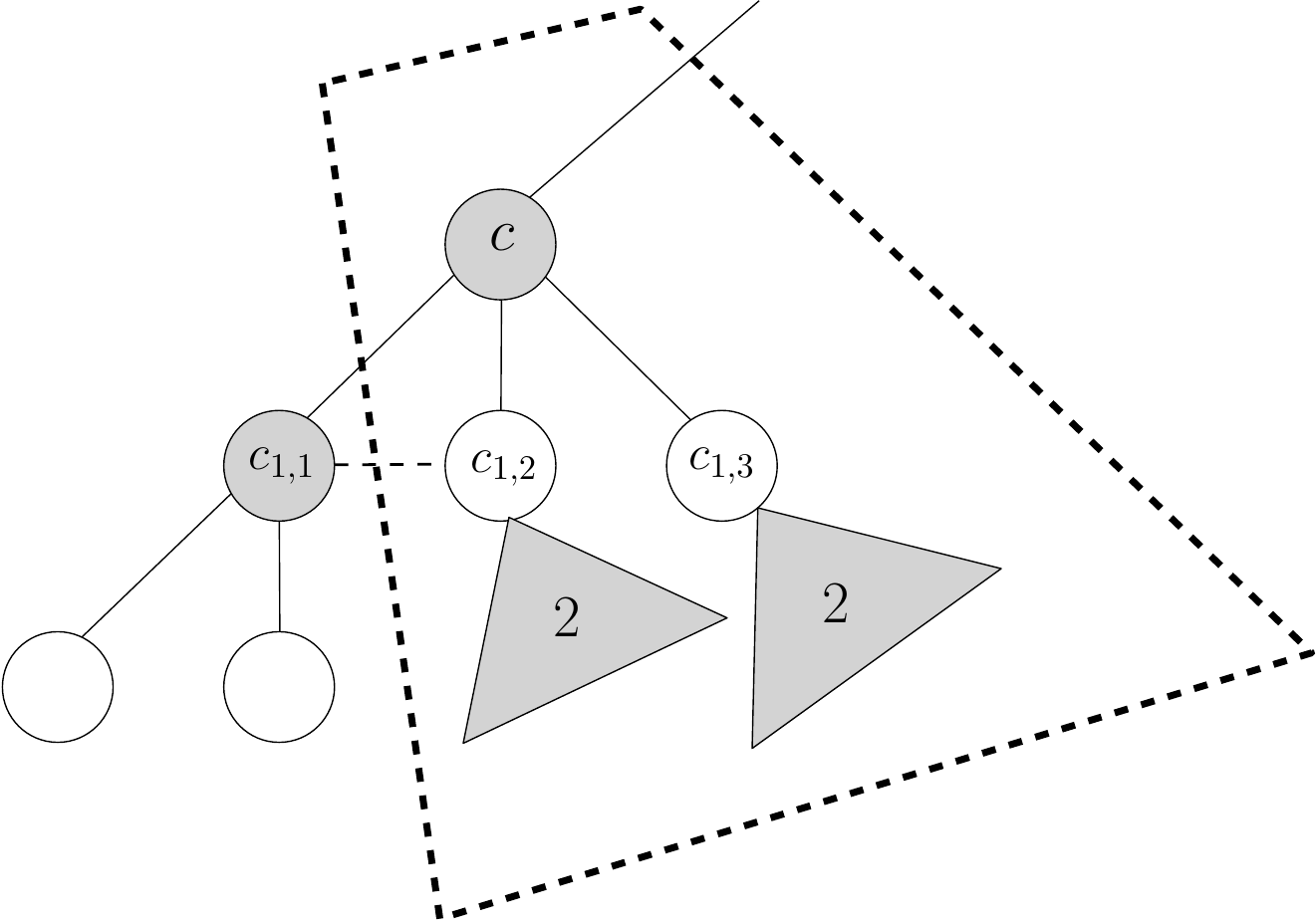}
\caption{The case described in Theorem~\ref{THM : LB Cactus} where $ALG$ does select $c_{1, 1}$.  The enclosed region contributes a performance of $\frac{5}{2}$.  A trap is continued in this case with the root $c_{1, 1}$. }\label{fig:cactus-second-accepted}
\end{figure}

\begin{figure}[h]
\centering
\includegraphics[scale=0.6]{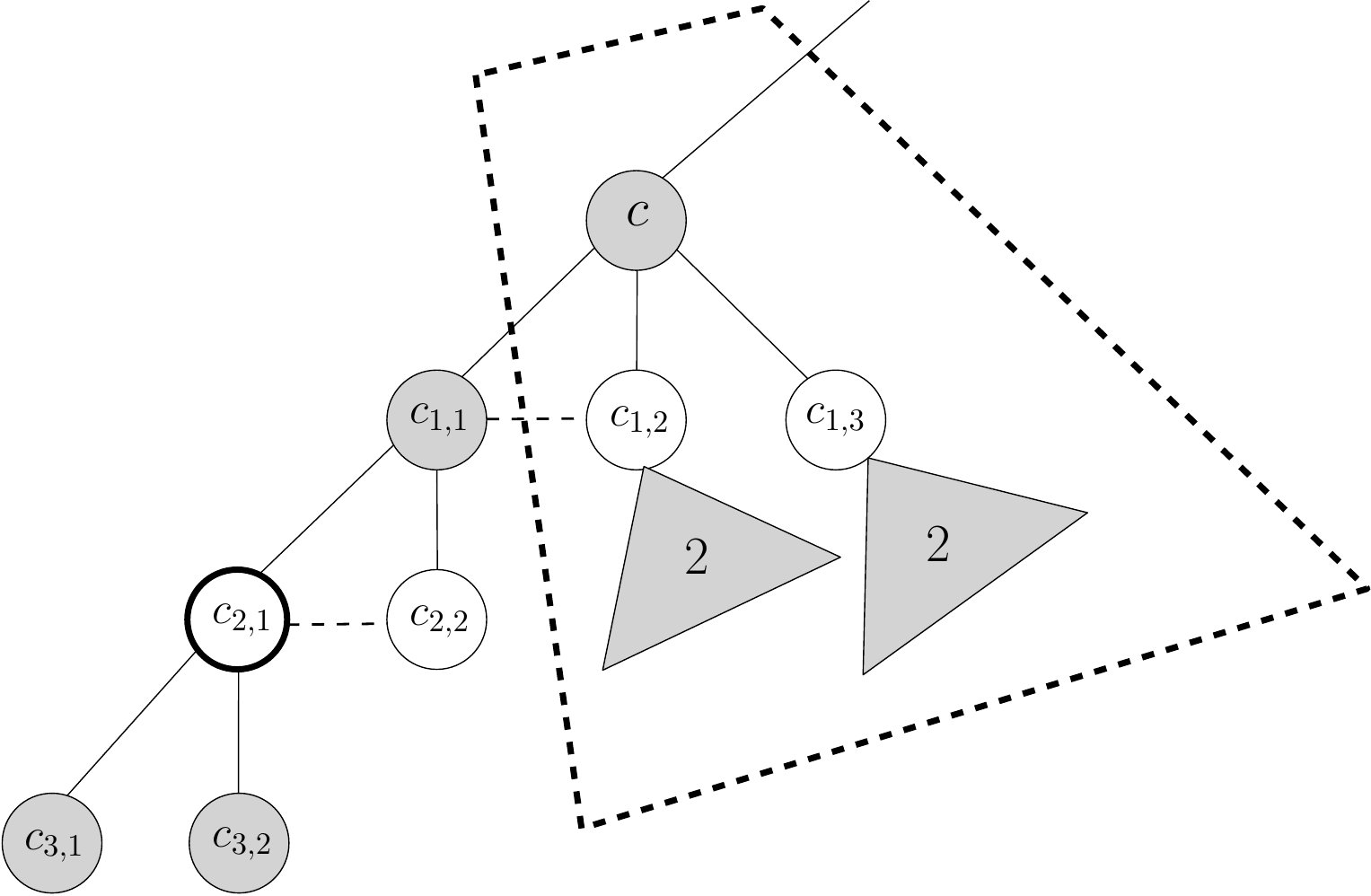}
\caption{The case described in Theorem~\ref{THM : LB Cactus} where $ALG$ does not select $c_{2, 1}$.}\label{fig:cactus-third-rejected}
\end{figure}

\begin{figure}[h]
\centering
\includegraphics[scale=0.6]{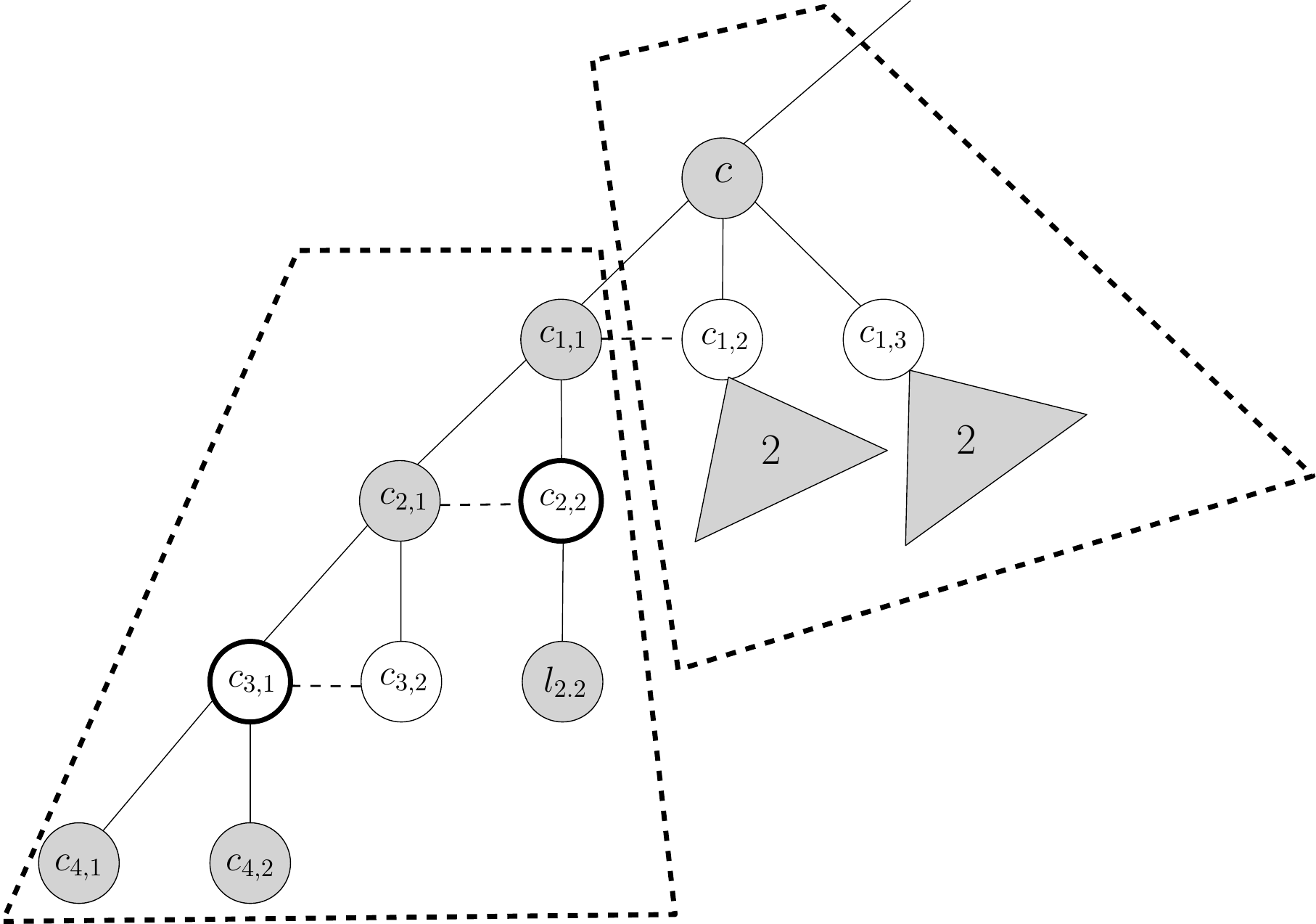}
\caption{The case described in Theorem~\ref{THM : LB Cactus} where $ALG$ does not select $c_{3, 1}$.  The enclosed regions each contribute a performance of $\frac{5}{2}$.}\label{fig:cactus-fourth-rejected}
\end{figure}

\begin{figure}[h]
\centering
\includegraphics[scale=0.6]{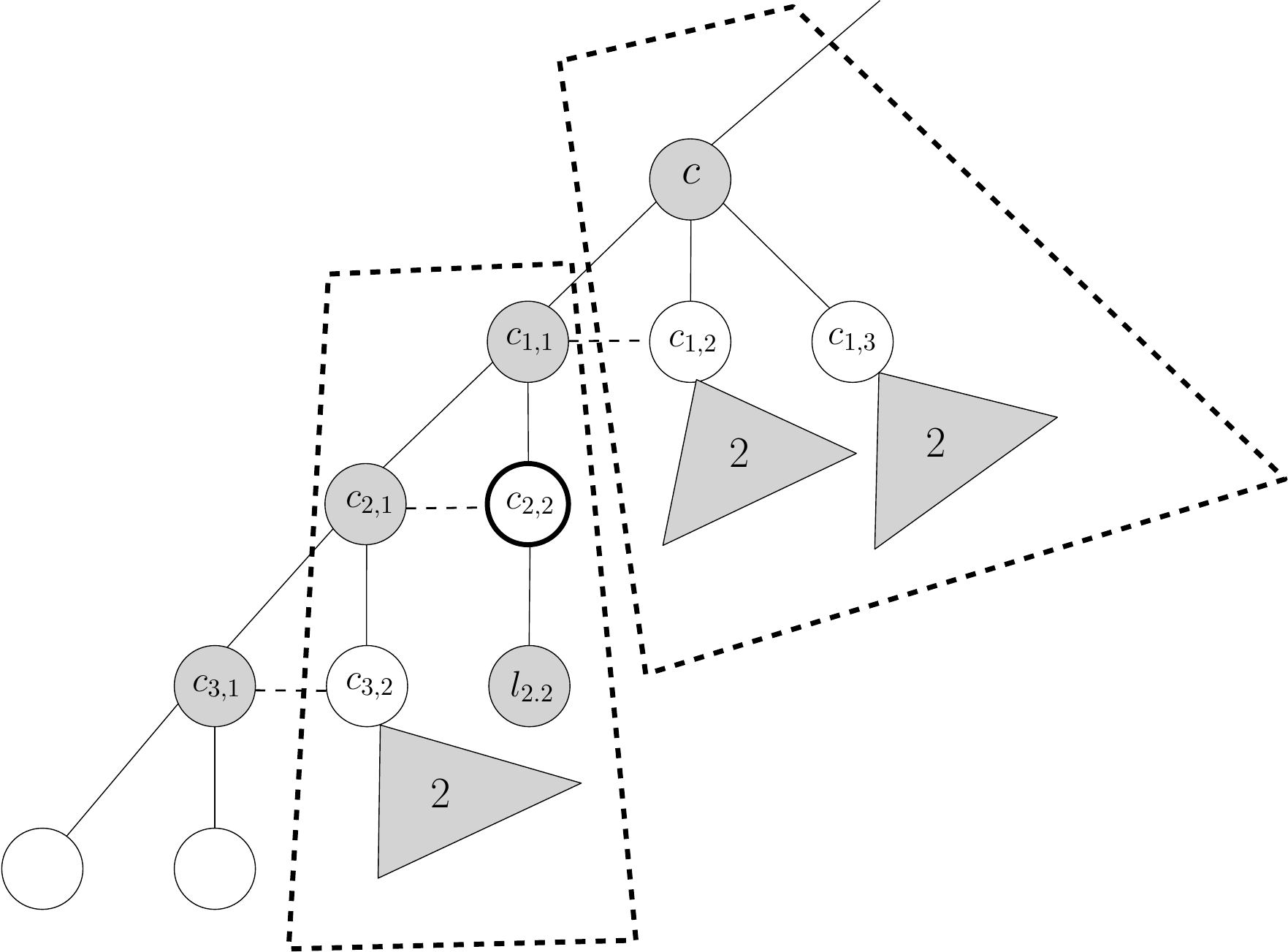}
\caption{The case described in Theorem~\ref{THM : LB Cactus} where $ALG$ does select $c_{3, 1}$.  The enclosed regions each contribute a performance of $\frac{5}{2}$.  The trap used on a selected root $c_{1, 1}$ is repeated on the root $c_{3, 1}$.}\label{fig:cactus-fourth-accepted}
\end{figure}

\begin{theorem}\label{thm: 2domUB}
$\rho($2-DOMINATE, CACTUS$) = \frac{5}{2}$.
\end{theorem}

The proof can be viewed as an adaptation of our proof for trees to cactus graphs. We use a charging argument similar to the one given in the section on trees.  Initially, a charge of $1$ is given for each $v \in S$, the charge on each vertex is then spread to certain neighbours, and we then show that $\sum\limits_{x \in N[v_i]}ch(x) \leq \frac{5}{2}$ for each $v_i \in OPT$.  We spread the charge according to the same rule given in the preceding section and recall that Observation~\ref{obs: Charged Once} (each vertex receives a new charge from one other vertex) still holds.  In the analysis of how the charge gets reallocated, the structure of the underlying graph is of paramount importance. We begin with 
an analogue to Lemma~\ref{Lemma : Tree Cross Edges}.

\begin{lemma}\label{Cact Cross Edge}
In cactus graphs, there is at most one cross edge incident on any $v_i$. In particular, $v_i$ has at most $2$ neighbours before it is revealed.
\end{lemma}
\begin{proof}
Suppose that $v_i \neq v_1$ since the statement is clearly true for $v_i = v_1$. Suppose for the sake of deriving a contradiction that, at time $i - 1$, $v_i$ has three neighbours $v_h, v_{i_1}, v_{i_2}$ where $v_h$ is the parent of $v_i$ and $\{ v_i, v_{i_1} \}, \{ v_i, v_{i_2} \}$ are cross edges.  Notice that $v_{i_1}, v_{i_2}$ are both visible at time $i - 1$ as otherwise would imply that $\{ v_i, v_{i_1} \}$ were a tree edge.  Thus, at time $i - 1$, $v_{i_1}$ is visible and there is only one tree edge incident on $v_i$.  In particular, this implies that there is a path consisting entirely of tree edges from $v_{i_1}$ to $v_h$ where said path does not contain the edge $\{ v_h, v_i \}$ since it does not pass through $v_i$ nor does it contain the edges $\{ v_i, v_{i_1} \}, \{ v_i, v_{i_2} \}$ since they are cross edges.  Thus, by adding edges $\{ v_h, v_i \}, \{ v_i, v_{i_1} \}$ to this path we obtain a cycle (in the completely revealed input graph) that contains the edge $ \{ v_h, v_i \}$ but does not contain the edge $\{ v_i, v_{i_2} \}$.  A similar argument yields that there is a path consisting of tree edges from $v_{i_1}$ to $v_h$ that does not contain the edges $\{ v_h, v_i \}, \{ v_i, v_{i_1} \}, \{ v_i, v_{i_2} \}$ and hence by adding edges $\{ v_h, v_i \}, \{ v_i, v_{i_1} \}$ we obtain a cycle which contains the edge $\{ v_h, v_i \}$ but does not contain the edge $\{ v_i, v_{i_1} \}$.  That is, two distinct cycles that share the common edge $\{ v_h, v_i \}$, a contradiction.
\end{proof}

Since $v_i$ has at most $2$ neighbours before it is revealed then it has at least $deg(v_i) - 2$ children.  
The following is analogous to Corollary~\ref{cor:  DegThree} for trees.

\begin{corollary}\label{Cact DegFour}
If $deg(v_i) \geq 4$ then $v_i \in S$.
\end{corollary}

\begin{lemma}\label{lemma: osbervations}
\begin{enumerate}
    \item\label{Cact Common Neighbour} If $v_i$ and $v_j$ both have charge equal to $1$ then they share no common neighbours.
    \item\label{lemma: Cact Charge} If $v_i$ and $v_j$ both have charge equal to $1$ then they are not adjacent.
    \item\label{Cor : Neighbourhood Charge} For any vertex $v_i$, at most one vertex in $N[v_i]$ has charge $1$.
\end{enumerate}
\end{lemma}

\begin{proof}
\begin{enumerate}
    \item Follows identically to the proof of Lemma $\ref{Tree Common Neighbour}$.

\item First, note that $v_1$ cannot have charge $1$ on any input with at least $2$ vertices.  Therefore we safely assume that $1 < i < j$ such that both $v_i$ and $v_j$ have a parent.  We assume for the sake of deriving a contradiction that $v_i$ and $v_j$ are adjacent.

Now, since both $v_i$ and $v_j$ have charge $1$ it follows that they are both saved vertices.  We first argue that both $v_i, v_j \notin S$.  Notice that any saved vertex $v_k$ has the property that $|N[v_k] \cap S| = 1$.  Therefore, if we assume by way of contradiction that $v_i \in S$ we obtain that $N[v_i] \cap S = N[v_j] \cap S = \{ v_i \}$ and therefore $v_i$ saves itself and $v_j$.  This yields that $X_i = N[v_{i}] \cap U_{i}$ contains $v_i$ and $v_j$.  In particular, we have that $|X_{i}| \geq 2$ with $v_i, v_j \in X_{i}$ and therefore $v_i$ and $v_j$ receive charge no larger than $\frac{1}{2}$, a contradiction.  An identical argument will yield that $v_j \notin S$.

Thus, we assume that $v_i$ is saved by a neighbour $v_{i'}$ and $v_j$ is saved by a neighbour $v_{j'}$ where $i', j' \notin \{ i, j \}$.  Moreover, $i' \neq j'$ since $v_i$ and $v_j$ can share no common neighbours by part $\ref{Cact Common Neighbour}$.  Thus, we have that $i, j, i', j'$ are all distinct with $i < j < i'$ and $i < j < j'$ since $i' = \max \{ k \mid v_k \in N[v_i] \}$ and $j' = \max \{ k \mid v_k \in N[v_j] \}$.  As mentioned above $v_i$ must have a parent $v_{h}$ where $h < i < j < i'$.  Therefore, $deg(v_i) \geq 3$ and since $v_i \notin S$ it follows by Corollary $\ref{Cact DegFour}$ that $deg(v_i) = 3$.

We are now in the situation where $v_i, v_j \notin S$ and $v_i$ is incident on exactly $3$ edges $\{ v_h, v_i \}$, $\{ v_i, v_{j} \}$, $\{ v_i, v_{i'} \}$ where exactly one of the edges $\{ v_i, v_{j} \}, \{ v_i, v_{i'} \}$ is a tree edge (and the other a cross edge).  We finish the proof by examining the two cases where $\textbf{(1) :}$ $\{ v_i, v_{i'} \}$ is a tree edge or $\textbf{(2) :}$ $\{ v_i, v_{j} \}$ is a tree edge.

$\textbf{Case 1 :}$ Suppose $\{ v_i, v_{i'} \}$ is a tree edge so that $v_{i'}$ is a child of $v_i$.  Therefore, $v_{i'} \in C_{i} \subseteq N(v_i) \cap U_i$, that is, $v_{i'}$ is an undominated neighbour of $v_i$ when $v_i$ is revealed. Since $v_j$ is saved then by Observation $\ref{lemma: Saviour Lemma}$ it follows that $v_j \in N(v_i) \cap U_i$, that is, $v_j$ is also an undominated neighbour of $v_i$ when $v_i$ is revealed.  That is, both $v_{i'}, v_j \in N(v_i) \cap U_i$ implying that $|N(v_i) \cap U_i| \geq 2$ but $v_i \notin S$, a contradiction.

$\textbf{Case 2 :}$ Suppose $\{ v_i, v_{j} \}$ is a tree edge so that $v_j$ is a child of $v_i$.  First notice that $\{ v_{i}, v_{j} \}$ is the only tree edge incident on $v_j$.  Indeed, if there were a tree edge $\{ v_j, v_l \}$ then $v_{l}$ would be the child of $v_j$.  Since $v_i$ is saved we have $v_i \in N(v_j) \cap U_j$ by Observation $\ref{lemma: Saviour Lemma}$ implying that $|N(v_j) \cap U_j| \geq 2$ but $v_j \notin S$.  Thus, we are in the situation depicted in Figure \ref{fig:cactus-upper-bound-final-case} where $\{ v_i, v_j \}$ is the only tree edge incident on $v_j$ and by assumption $\{ v_h, v_i \}, \{ v_i, v_j \}$ are the only two tree edges incident on $v_i$.  Therefore we have a path from $v_{i'}$ to $v_h$ consisting of tree edges where said path does not contain the edges $\{ v_h, v_i \}, \{ v_i, v_{i'} \}, \{ v_i, v_j \}, \{ v_j, v_{j'} \}$.  Thus, by adding edges $\{ v_h, v_i \}, \{ v_i, v_i' \}$ to this path we obtain a cycle (in the completely revealed input) that contains the edge $\{ v_h, v_i \}$ but does not contain the edge $\{ v_j, v_{j'} \}$.  Similarly, there is a path from $v_{j'}$ to $v_h$ consisting of tree edges where said path does not contain the edges $\{ v_h, v_i \}, \{ v_i, v_{i'} \}, \{ v_i, v_j \}, \{ v_j, v_{j'} \}$ and by adding edges $\{ v_h, v_i \}, \{ v_i, v_j \}, \{ v_j, v_{j'} \}$ we obtain a cycle (in the completely revealed input) that contains the edge $\{ v_h, v_i \}$ but does not contain the edge $\{ v_i, v_{i'} \}$.  That is, two distinct cycles that share the common edge $\{ v_h, v_i \}$, a contradiction.

\item Follows immediately from the previous parts.
\end{enumerate}
\end{proof}

\begin{figure}[h]
\centering
\includegraphics[scale=0.6]{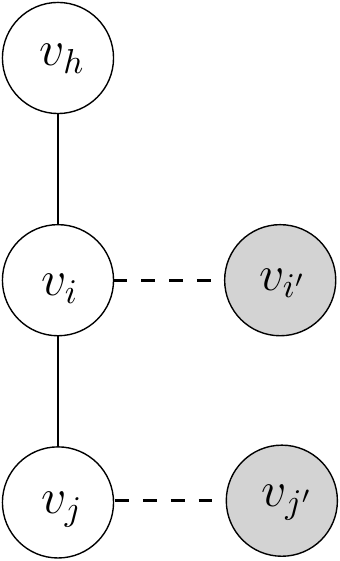}
\caption{Case $2$ of the second part of Lemma~\ref{lemma: osbervations}.}\label{fig:cactus-upper-bound-final-case}
\end{figure}

\begin{figure}[h]
\centering
\includegraphics[scale=0.6]{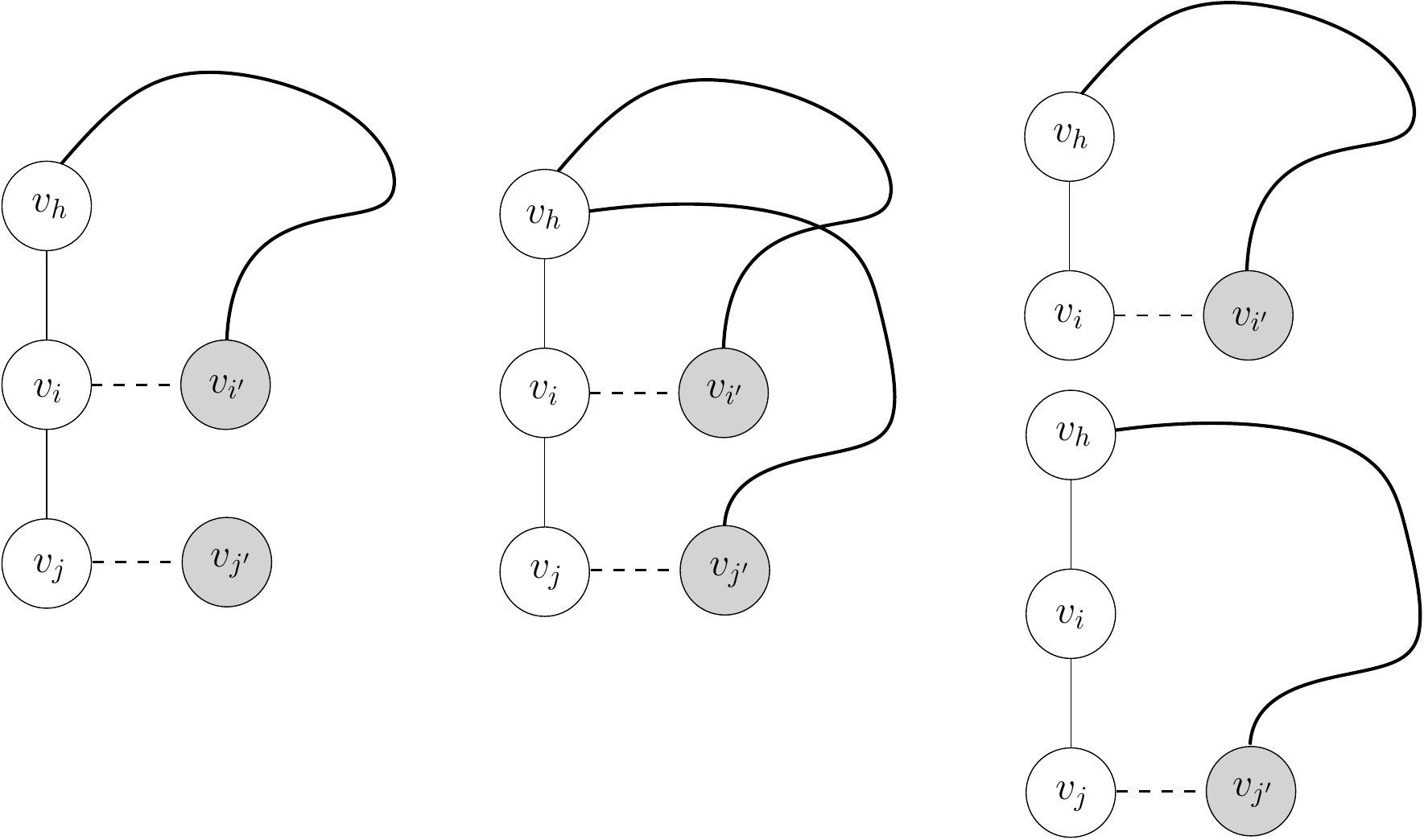}
\caption{Resolution of the preceding case in Figure \ref{fig:cactus-upper-bound-final-case}.  Two cycles sharing the common edge $\{ v_h, v_i \}$.}\label{fig:cactus-upper-bound-final-case-resolution}
\end{figure}

Now, we are ready to prove the upper bound for Theorem~\ref{thm: 2domUB}.

\begin{proof}[Proof of Theorem~\ref{thm: 2domUB}]
The lower bound follows from Theorem \ref{THM : LB Cactus}.  Let $v_i \in OPT$ be an arbitrary vertex in $OPT$.  We consider two cases $\textbf{(1)}$ $deg(v_i) \leq 3$ or $\textbf{(2)}$ $deg(v_i) \geq 4$.

$\textbf{Case 1 :}$ Suppose that $deg(v_i) \leq 3$ and hence $|N[v_i]| \leq 4$.  By Lemma~\ref{lemma: osbervations} part $\ref{Cor : Neighbourhood Charge}$ it follows that at most one vertex in $N[v_i]$ has charge $1$.  If no vertices in $N[v_i]$ have charge $1$ then $ch(x) \leq \frac{1}{2}$ for each $x \in N[v_i]$ and we obtain that $\sum\limits_{x \in N[v_i]}ch(x) \leq 4\big(\frac{1}{2}\big) = 2 < \frac{5}{2}$.  If there is exactly one vertex $x' \in N[v_i]$ with charge $1$ we therefore obtain that $\sum\limits_{x \in N[v_i]}ch(x) = \sum\limits_{x \in N[v_i] \setminus \{ x' \}}ch(x) + ch(x') \leq \frac{3}{2} + 1 = \frac{5}{2}$.

$\textbf{Case 2 :}$ Suppose that $deg(v_i) \geq 4$.  By Corollary \ref{Cact DegFour} it follows that $v_i \in S$ with at least $2$ children.  Let $C_i = V_i \setminus V_{i - 1}$ denote the children of $v_i$ and remark that $C_i \subseteq X_i$.  That is, each child of $v_i$ is charged by $v_i$ and only $v_i$.  Therefore the children of $v_i$ can receive at most the full initial charge on $v_i$ and thus attribute a charge of at most $1$.  

Now we claim that any vertex in $N[v_i] \setminus C_i$ has a charge of at most $\frac{1}{2}$.  Indeed, suppose a vertex $v_{i'} \in N[v_i] \setminus C_i$ has charge $1$ then it must be saved by $v_i$ since $|N[v_{i'}] \cap S| = 1$ for any saved vertex $v_{i'}$.  That is, there is exactly one vertex in its closed neighbourhood that is selected and since $v_i$ is selected it must be $v_i$.  Thus, we must have that $v_{i'} \in X_i$ but since $C_i \subseteq X_i$ we know that $|X_i| \geq 2$ and thus $v_{i'}$ receives a charge of no more than $\frac{1}{2} < 1$, contradicting our assumption that $v_{i'}$ has charge $1$. Thus, by remarking that $|N[v_i] \setminus C_i| \leq 3$ we obtain that $\sum\limits_{x \in N[v_i]}ch(x) = \sum\limits_{v_{j} \in C_i}ch(v_j) + \sum\limits_{v_{i'} \in N[v_i] \setminus C_i}ch(v_{i'}) \leq 1 + 3\big(\frac{1}{2}\big) = \frac{5}{2}$ as desired.

\end{proof} 

\subsection{Graphs of Bounded Degree}
\label{ssec:degree}

We study the problem when the inputs are restricted to graphs of bounded degree.  That is, a positive integer $\Delta \geq 2$ is provided to the algorithm beforehand and the adversary is restricted to presenting graphs where every vertex has degree no larger than $\Delta$.  The problem of bounded degree graphs was explored in \cite{OnlineDominatingSet3} although within the vertex arrival model described earlier.  The authors show that a greedy strategy obtains a competitive ratio no larger than $\Delta$ and, when inputs are further restricted to be ``always-connected'' (i.e. each prefix of the input is connected) they provide a lower bound of $\Delta - 2$ for any algorithm.  

By definition, any input belonging to our setting is ``always-connected'' yet the lower bound of $\Delta - 2$ does not apply.  In particular, we show that $\big \lceil \sqrt{\Delta} \big \rceil$-DOMINATE is $3\sqrt{\Delta}$-competitive along with a lower bound of $\Omega(\sqrt{\Delta})$ for any online algorithm, essentially closing the problem in our setting.  As previously mentioned, the authors in \cite{OnlineDominatingSet1} consider a setting similar to ours where their adversary is not required to reveal visible vertices and they assume that an algorithm has additional knowledge of input size $n$.  In this setting they provide an algorithm that achieves competitive ratio of $\Theta(\sqrt{n})$ for arbitrary graphs.  For the upper bound below we follow a proof nearly identical to theirs modulo some minor details and definitions.  

\begin{definition}
A vertex $v_i \in S$ is said to be heavy if $|N(v_i) \cap U_i| \geq \big \lceil \sqrt{\Delta} \big \rceil$ and light otherwise.  We let $H$ and $L$ denote the set of heavy and light vertices in $S$ so that $|S| = |H| + |L|$.
\end{definition}

To establish that $\big \lceil \sqrt{\Delta} \big \rceil$-DOMINATE is $3\sqrt{\Delta}$-competitive we use a charging argument, but it is quite different from the arguments in Sections~\ref{ssec:trees} and~\ref{ssec:cactus}.  Initially, let $ch(v) = 1$ for each $v \in S$ so that $|S| = \sum\limits_{v \in S}ch(v)$.  Then spread the charge from $S$ strictly to vertices in $OPT$ so that $\sum\limits_{v \in S}ch(v) = \sum\limits_{v \in OPT}ch^{*}(v)$ where $ch^{*}(v)$ is the new charge on a vertex in $OPT$.  We then show that $ch^{*}(v) \leq 2\sqrt{\Delta}$ for all $v \in OPT$ and thus $|S| = \sum\limits_{v \in S}ch(v) = \sum\limits_{v \in OPT}ch^{*}(v) \leq |OPT|2\sqrt{\Delta}$ and the result then follows.  We spread the charge from $S$ to $OPT$ according to the following rules:

\begin{enumerate}
    \item If $v_i \in S \cap OPT$ then $v_i$ keeps its full initial charge. 
    
    \item If $v_i \in H \setminus OPT$ then its spread its initial charge evenly over all vertices in $OPT$.  That is, each $v \in OPT$ obtains an additional charge of $\frac{1}{|OPT|}$ from $v_i$.
    
    \item For each $v_i \in L \setminus OPT$, let $s(v_i)$ denote the set of vertices saved by $v_i$.  Given a vertex $v_{i'} \in s(v_i)$ let $opt(v_{i'}) = v_{i'}$ if $v_{i'} \in OPT$ and $opt(v_{i'}) = \min \{ k \mid v_k \in N(v_{i'}) \cap OPT \}$ otherwise.  For each $v_{i'} \in s(v_i)$, $v_i$ spreads $\frac{1}{|s(v_i)|}$ to $opt(v_{i'})$.
\end{enumerate}

\begin{lemma}\label{lightChargeLemma}
If $v_i \in OPT$ then it receives charge from at most $\lceil \sqrt{\Delta} \rceil$ light vertices.  
\end{lemma}
\begin{proof}
We consider two cases; $\textbf{(1)}$ $v_i \in S$ or $\textbf{(2)}$ $v_i \notin S$.

$\textbf{Case 1 : }$ Suppose that $v_i \in S$, we show that $v_i$ then it receives no charge from a distinct light vertex (therefore it receives charge from at most one light vertex, itself).  Since $v_i \in S$ this implies that it is not saved by any $v_j, j \neq i$.  Thus, if $v_i$ were to receive charge from a light vertex it must be that $v_i = opt(v_{i'})$ for some $v_{i'}$ that is saved by some $v_k \in L$ different from $v_i$.  More precisely, $v_i$ must be adjacent to some $v_{i'}$ that is saved by some $v_k$ with $k \neq i$.  Yet, if $v_{i'} \in N(v_i)$ is saved then $N[v_{i'}] \cap S = \{ v_i \}$ so this cannot be the case.

$\textbf{Case 2 : }$ Assume that $v_i \notin S$ and first remark that $v_i$ is saved by at most one vertex so that it receives at most one charge from a light vertex in this way.  If $v_i$ receives charge from any other light vertex $v_{k} \in L$, it must be that $v_i$ is adjacent to some vertex $v_{i'}$ that is saved by $v_k$.  By Observation $\ref{lemma: Saviour Lemma}$ it must be that $v_{i'} \in N(v_i) \cap U_i$, that is, is undominated when $v_i$ is revealed.  All this to say, that any light vertex that charges $v_i$ determines at least one neighbor of $v_i$ that is undominated at time $i$.  Since $v_i \notin S$ we have $|N(v_i \cap U_i| \leq \lceil \sqrt{\Delta} \rceil - 1$ and thus accounting for possibly one light vertex that charges $v_i$ there are at most $\lceil \sqrt{\Delta} \rceil$ light vertices that charge $v_i$.
\end{proof}

\begin{lemma}\label{lemma : HeavyCharge}
$\frac{|H|}{|OPT|} \leq \sqrt{\Delta} + \frac{1}{\sqrt{\Delta}}$.
\end{lemma}
\begin{proof}
Since every vertex in $H$ is selected because it dominated at least $\lceil \sqrt{\Delta} \rceil$ undominated vertices it follows that $|H| \leq \big \lfloor \frac{n}{\lceil \sqrt{\Delta} \rceil} \big \rfloor$.  Moreover, by a standard result, first proved by Berge~\cite{Berge62}, a lower bound on $OPT$ is $|OPT| \geq \big \lceil \frac{n}{\Delta + 1} \big \rceil$.  Ultimately this yields that $$ \frac{|H|}{|OPT|} \leq 
\frac{\big \lfloor \frac{n}{\lceil \sqrt{\Delta} \rceil} \big \rfloor}{\big \lceil \frac{n}{\Delta + 1} \big \rceil} \leq 
\frac{\frac{n}{\lceil \sqrt{\Delta} \rceil}}{\frac{n}{\Delta + 1}} \leq 
\frac{\frac{n}{\sqrt{\Delta}}}{\frac{n}{\Delta + 1}} = 
\frac{\Delta + 1}{\sqrt{\Delta}} = 
\sqrt{\Delta} + \frac{1}{\sqrt{\Delta}}.$$
\end{proof}

\begin{theorem}
$\rho(\big \lceil \sqrt{\Delta} \big \rceil$-DOMINATE$, \Delta$-BOUNDED$) \leq 3\sqrt{\Delta}$.
\end{theorem}
\begin{proof}
Consider an arbitrary vertex $v_i \in OPT$.  In light of Lemma $\ref{lightChargeLemma}$ we see that it receives charge from at most $\lceil \sqrt{\Delta} \rceil$ light vertices, where each charge is no larger than $1$.  Moreover, by Lemma $\ref{lemma : HeavyCharge}$ the charge received by the heavy vertices is at most $\sqrt{\Delta} + \frac{1}{\sqrt{\Delta}}$ and $v_i$ possibly receives charge from itself (it may be a heavy or light vertex).  In particular we obtain that $$ ch(v_i) \leq \frac{|H|}{|OPT|} + \big\lceil \sqrt{\Delta} \big\rceil + 1 \leq \big(\sqrt{\Delta} + \frac{1}{\sqrt{\Delta}}\big) + \big\lceil \sqrt{\Delta} \big\rceil + 1 \leq 3\sqrt{\Delta}.$$
\end{proof}

We now prove a lower bound $\Omega(\sqrt{\Delta})$ for any online algorithm.  We should note that the adversarial input is bounded in size by a function of $\Delta$.   Although we have omitted the details, it is straightforward to extend the input so that the lower bound is in fact an asymptotic one. 

\begin{theorem}\label{thm: LB Delta}
$\rho(ALG, \Delta$-BOUNDED$) = \Omega(\sqrt{\Delta}).$
\end{theorem}
\begin{proof}
For simplicity we assume that $\Delta$ is a perfect square.  Reveal $v_1$ with $\Delta$ children and reveal each child of $v_1$ with an additional $\sqrt{\Delta}$ children.  Of the $\Delta$ children of $v_1$, suppose that $ALG$ selects exactly $j$ where $0 \leq j \leq \Delta$.  For the $\Delta - j$ vertices not selected, their $\sqrt{\Delta}$ neighbours are revealed to have degree $1$ and $ALG$ is forced to select each of these $(\Delta - j)(\sqrt{\Delta})$ vertices of degree $1$.  

Let $S_j$ denote the set of the $j$ selected vertices in $N(v_1)$ and $X = \bigcup_{v_i \in S_j}N(v_i)$.  Since each vertex in $S_j$ has $\sqrt{\Delta}$ children, it follows that $|X| = j\sqrt{\Delta}$.  Partition the vertices of $X$ into $\lceil \frac{j\sqrt{\Delta}}{\Delta} \rceil = \lceil \frac{j}{\sqrt{\Delta}} \rceil$ parts of size $\Delta$ (with at most one part having size $< \Delta$).  Letting the parts be $X_1, X_2, ..., X_{\lceil \frac{j}{\sqrt{\Delta}} \rceil}$ we reveal each vertex in a given part to a common vertex $y_i$ (see figure \ref{fig:delta-bounded-lb-1} for an example).  $ALG$ must select at least one vertex for each part to dominate $y_i$ and therefore at least an additional $\lceil \frac{j}{\sqrt{\Delta}} \rceil$ vertices are selected.  

In total, $ALG$ selects at least $j + (\Delta - j)(\sqrt{\Delta}) + \frac{j}{\sqrt{\Delta}}$ whereas $OPT$ simply selects $v_1$, the $\Delta - j$ vertices in $N(v_1) \setminus S_j$ and the $\frac{j}{\sqrt{\Delta}}$ vertices with labels $y_i$.  Ultimately we have 
\begin{align*}
    \frac{ALG}{OPT} &\geq \frac{j + (\Delta - j)(\sqrt{\Delta}) + \frac{j}{\sqrt{\Delta}}}{1 + (\Delta - j) + \frac{j}{\sqrt{\Delta}}} = \frac{j + j \sqrt{\Delta} + (\Delta-j) \Delta} {j + \sqrt{\Delta}+(\Delta-j)\sqrt{\Delta}}\\
    &=\frac{\sqrt{\Delta}(j/\sqrt{\Delta} + j + (\Delta-j)\sqrt{\Delta})}{2(j/2 + \sqrt{\Delta}/2 + (\Delta-j)\sqrt{\Delta}/2)} \ge \frac{\sqrt{\Delta}}{2},
\end{align*}
where the last inequality follows from the fact that $j/2+\sqrt{\Delta}/2 + (\Delta-j)\sqrt{\Delta}/2 \le j/\sqrt{\Delta}+j + (\Delta-j)\sqrt{\Delta}$, since $\sqrt{\Delta}/2 \le j/\sqrt{\Delta}+j/2 + (\Delta-j)\sqrt{\Delta}/2$, which can be seen since when $j < \Delta$ then the last term on the right hand side already is at least as large as the left hand side and when $j = \Delta$ then the middle term on the right hand side is at least the left hand side.

\end{proof}

\begin{figure}[h]
\centering
\includegraphics[scale=0.6]{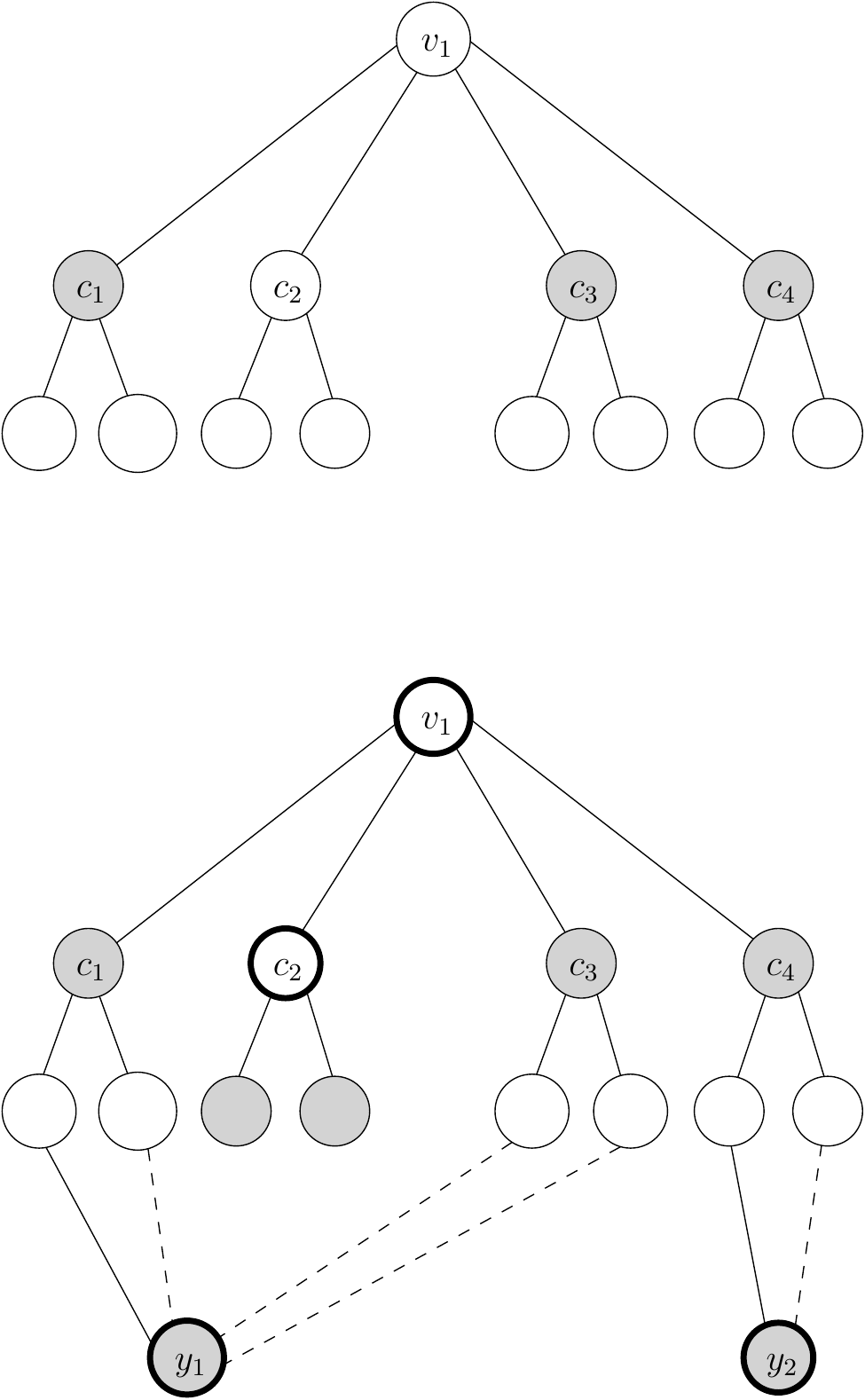}
\caption{An instance described in the proof of Theorem~\ref{thm: LB Delta} with $\Delta = 4$.  The top depicts the graph after the children of $v_1$ have been revealed.  Assuming that $ALG$ selects $\{ c_1, c_3, c_4 \}$ above, the bottom depicts the completely revealed graph.}\label{fig:delta-bounded-lb-1}
\end{figure}

\subsection{Graphs with Bounded Claws}
\label{ssec:claws}

Let $t \geq 3$, a graph $G$ is said to be $K_{1, t}$-free if it contains no induced subgraph isomorphic to $K_{1, t}$.  When $t = 3$, this is the well-studied class of claw-free graphs.  In this section we study $K_{1,t}$-free graphs, which we also refer to as graphs with bounded ``claws''.

From the preceding sections one might notice that the existence of an induced subgraph $K_{1, t}$ poses challenges for an algorithm.  This section suggests that this intuition holds more than just a grain of truth.  We show that, when inputs are restricted to $K_{1, t}$-free graphs, the competitive ratio of every algorithm is bounded below by $t - 1$ and there is an algorithm that achieves competitive ratio $t - 1$.  The upper bounds that we have demonstrated so far were all based on the $k$-DOMINATE algorithm for a suitable choice of parameter $k$. Interestingly, our upper bound on $K_{1,t}$-free graphs is based on a conceptually simpler GREEDY algorithm. The analysis is no longer based on a charging scheme, but follows from combinatorial properties of graphs with bounded claws. 

\begin{theorem}\label{thm : LB T-Claw-Free Graphs}
$\rho(ALG, K_{1, t}$-FREE$) \geq t - 1.$
\end{theorem}
\begin{proof}
Reveal $v_1$ with $t - 1$ children.  If $ALG$ does not select $v_1$ then the input terminates as a star on $t$ vertices (i.e. the $t - 1$ neighbours of $v_1$ are revealed with no additional neighbours).  Any feasible algorithm must select the $t - 1$ neighbours of $v_1$ whereas $OPT$ selects $v_1$ and the statement then follows.  Suppose that $ALG$ selects $v_1$ and let $c_i, 1 \leq i \leq t - 1$ be the children of $v_1$.  Reveal $c_1$ as adjacent to each child of $v_1$ and with an additional $t - 2$ children. If $ALG$ does not select $c_1$ then the children of $c_1$ are revealed as leaves whereas the rest of the input is revealed to be a clique.  That is, $N[v_1]$ is a clique and only $c_1$ has children.  $ALG$ selected $v_1$ and is forced to select the $t - 2$ children of $c_1$ whereas $OPT$ selects only $c_1$ as a single dominating vertex.  It is not hard to see that this input is $K_{1, t}$-free and the result then follows (see Figure~\ref{fig:bounded-claw-first-rejected} for an example). 

Suppose that $ALG$ selects $c_1$, the input then continues in the following way; For each $2 \leq j \leq t - 2$, (as long as $ALG$ is accepting $c_j$) we reveal $c_j$ as adjacent to every visible vertex and with an additional $t - 3$ children.  That is, $c_j$ is adjacent to each child $c_i, i \neq j$ of $v_1$ and the grandchildren of $v_1$ (i.e. the children of all the $c_i$ with $1 \leq i \leq j$) so that $c_j$ is a single dominating vertex of this prefix.  

$\textbf{Case 1 :}$ If there is some $2 \leq j \leq t - 2$ such that $ALG$ does not select $c_j$ then the $t - 3$ children of $c_j$ are revealed as leaves, $N[v_i]$ is revealed as a clique, and the $(t - 2) + \sum\limits_{i = 2}^{j}(t - 3) = j(t - 3) + 1$ grandchildren of $v_1$ are revealed to form a clique.  At this point, $ALG$ has selected $\{ v_1, c_1, ..., c_{j - 1} \}$ and is now forced to select the $t - 3$ children of $c_j$ for an output of at least $j + (t - 3) \geq 2 + (t - 3) = t - 1$ whereas $OPT$ selects only $c_j$ so that $\frac{ALG}{OPT} \geq \frac{t - 1}{1}$.  

We now argue that this input is $K_{1, t}$-free.  Notice that for all $v$ in this input we have $N(v) \subseteq N(c_j)$ so that if there is a an induced $K_{1, t}$ with central vertex $v$ then there is a claw with central vertex $c_j$.  Therefore it is sufficient to show that is no claw with central vertex $c_j$ to finish the claim.  Suppose for contradiction's sake that there were an induced $K_{1, t}$ where $c_j$ is the central vertex and the $t$ neighbors of $c_j$ are all pairwise non-adjacent.  Let $G$ denote the grandchildren of $v_1$ and remark that any neighbor of $c_j$ is either a child of $c_j$, a grandchild of $v_1$, or a vertex from $N[v_1] \setminus \{ c_j \}$.  Since there are $t$ vertices and $c_j$ only has $t - 3$ children by the pigeonhole principle we must have at least two vertices $u, v$ that both are grandchildren of $v_1$ or both belong $N[v_1] \setminus \{ c_j \}$.  Yet, both the set of grandchildren of $v_1$ and $N[v_1] \setminus \{ c_j \}$ are cliques.  Therefore we have that $u$ and $v$ are adjacent, contradicting our assumption.  

$\textbf{Case 2 :}$ If $ALG$ selects each $c_i, 1 \leq i \leq t - 2$ then the $(t - 2)(t - 3) + 1$ grandchildren of $v_1$ are then revealed to form a clique ($N[v_1]$ has already been revealed as a clique).  $ALG$ has already selected $\{ v_1, c_1, ..., c_{t-2} \}$ and therefore has an output of at least $t - 1$ whereas $OPT$ selects only $c_{t - 2}$.  An argument similar to the one above will yield that this input is $K_{1, t}$-free and the result then follows.
\end{proof}

\begin{figure}[h]
\centering
\includegraphics[scale=0.6]{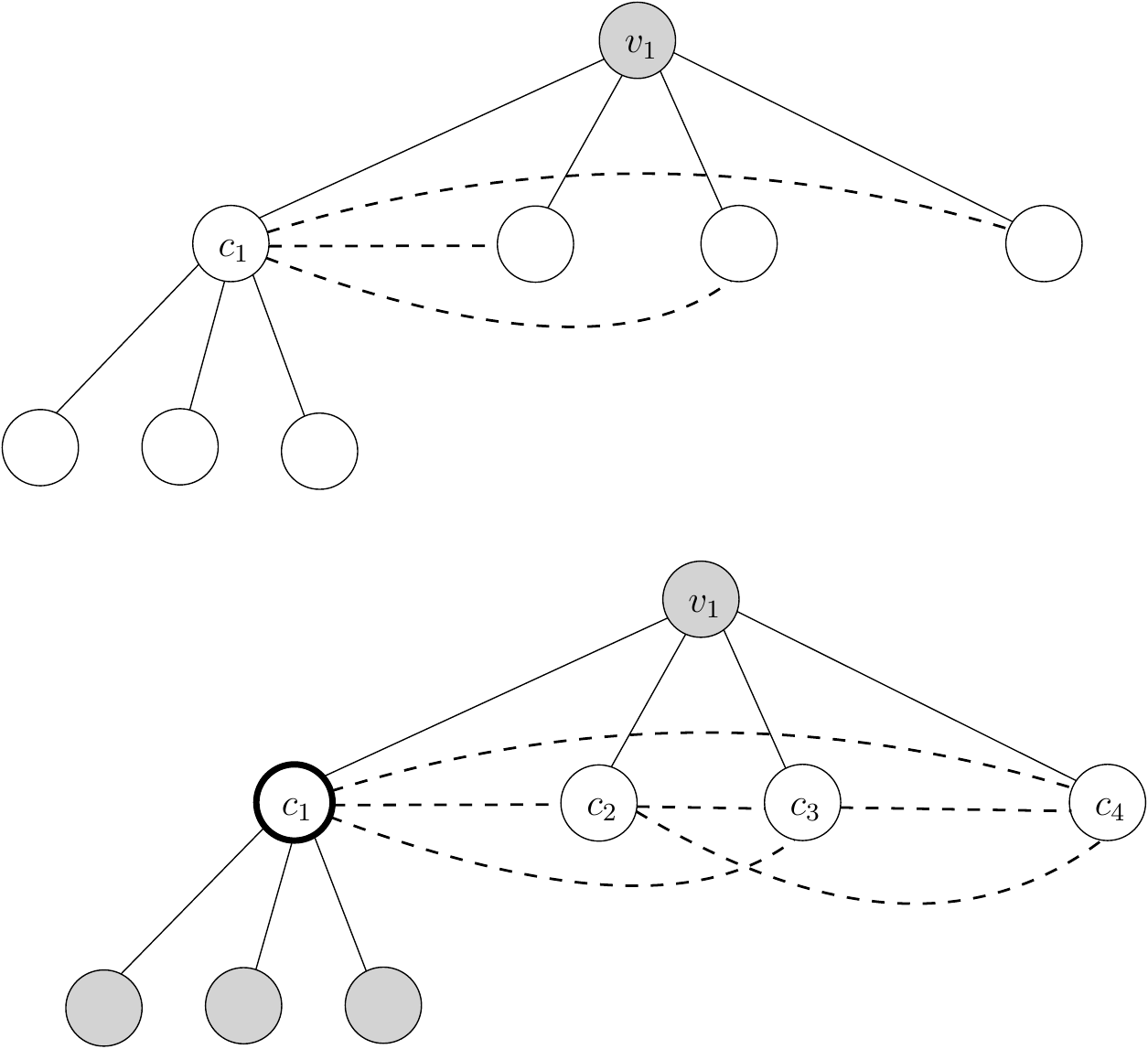}
\caption{An instance described in Theorem~\ref{thm : LB T-Claw-Free Graphs}  with $t = 5$ where ALG does not select $c_1$.  The top depicts the graph at the moment $c_1$ was revealed and the bottom depicts the completely revealed graph.}\label{fig:bounded-claw-first-rejected}
\end{figure}

\begin{figure}[h]
\centering
\includegraphics[scale=0.6]{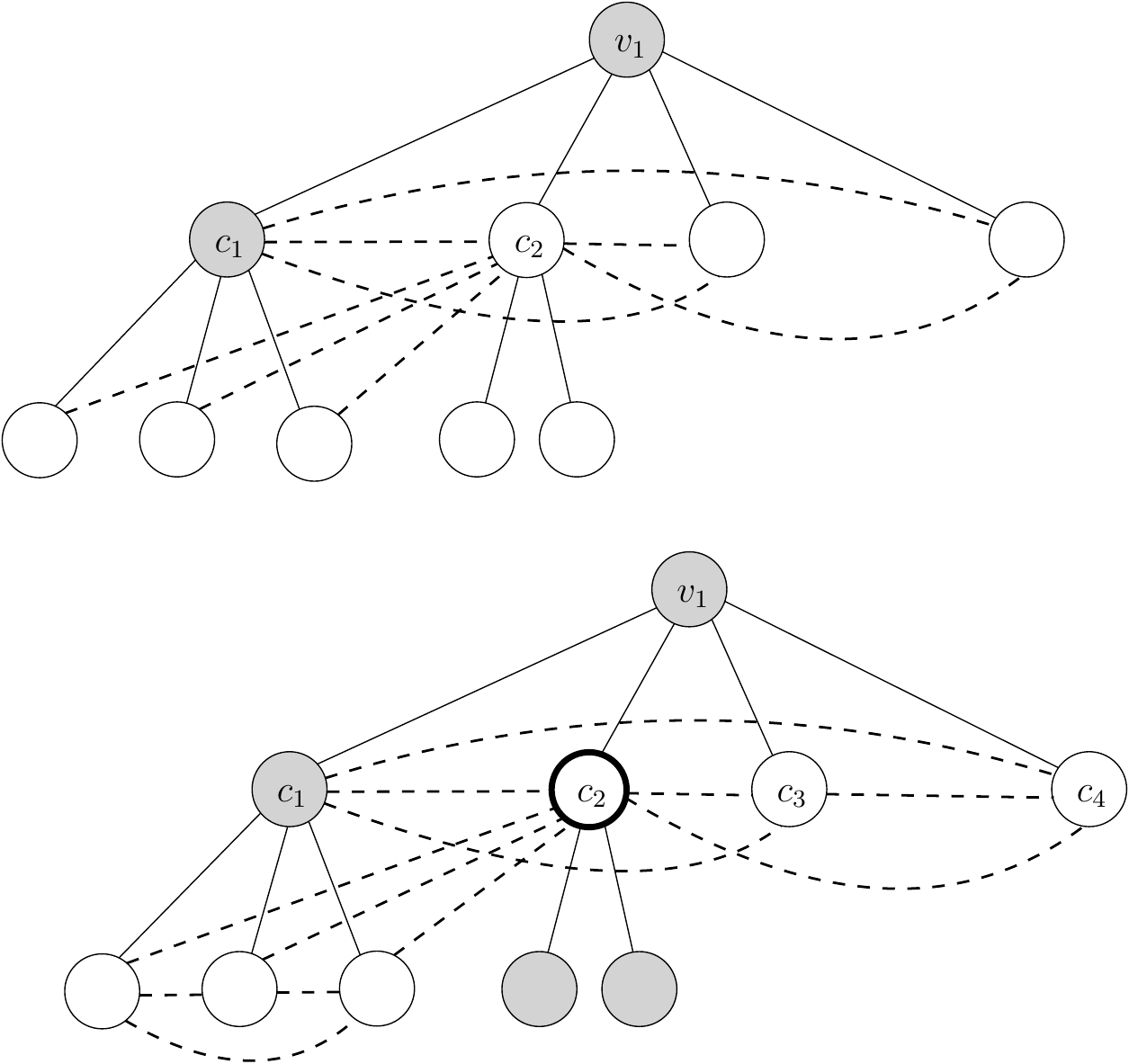}
\caption{An instance described in Theorem~\ref{thm : LB T-Claw-Free Graphs}  with $t = 5$ where ALG does not select $c_2$.  The top depicts the graph at the moment $c_1$ was revealed and the bottom depicts the completely revealed graph.}\label{fig:bounded-claw-second-rejected}
\end{figure}

\begin{figure}[h]
\centering
\includegraphics[scale=0.6]{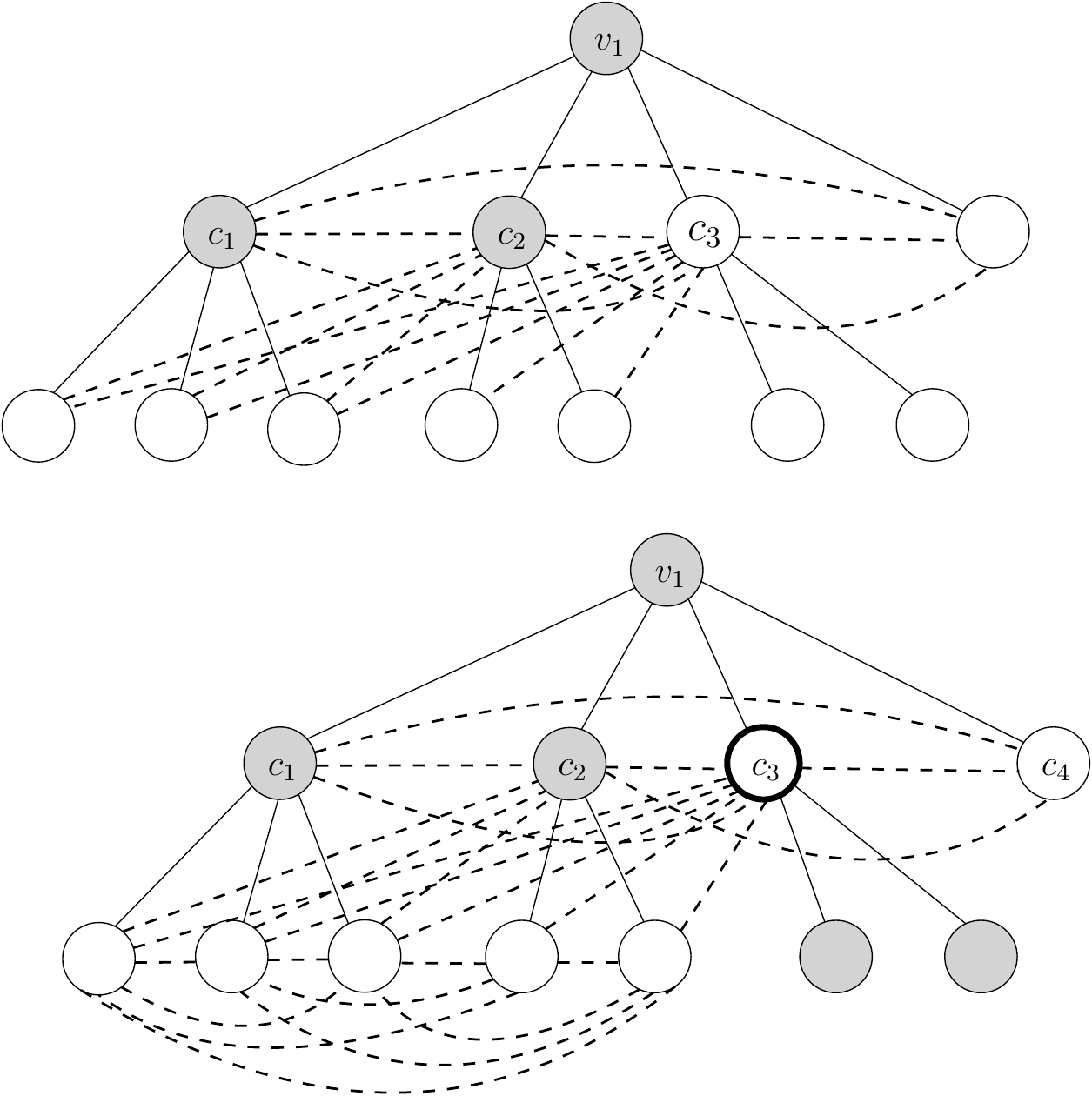}
\caption{An instance described in Theorem~\ref{thm : LB T-Claw-Free Graphs}  with $t = 5$ where ALG does not select $c_3$.  The top depicts the graph at the moment $c_3$ was revealed and the bottom depicts the completely revealed graph.}\label{fig:bounded-claw-third-rejected}
\end{figure}

When inputs are restricted to $K_{1, t}$-free graphs, we show that the online algorithm GREEDY is $(t-1)$-competitive.  The crucial observation to make here is that the output of GREEDY is an independent set.  We provide a result below that is a straightforward generalization of one given in \cite{FPTClawFree}.  The simplicity of the result suggests that it may have appeared in earlier work.  

\begin{lemma}\label{lemma : I.S. Claw-Free}
Let $t \geq 3$, $G = (V, E)$ be a $K_{1, t}$-free graph and $I$ be any independent set in $G$.  Then $|D| \geq \frac{|I|}{t - 1}$ for any dominating set $D$ in $G$.
\end{lemma}
\begin{proof}
Suppose for the sake of deriving a contradiction that there is some dominating set $D$ in $G$ with $|D| < \frac{|I|}{t - 1}$.  Remarking that the vertices of $D$ dominate the vertices of $I$ as $D$ is a dominating set we notice that there is some vertex $v \in D$ that dominates at least $t$ vertices of $I$ (i.e. if every vertex of $D$ dominated at most $t - 1$ vertices then $D$ would dominate at most $(t - 1)|D| < |I|$ vertices).  Moreover, since $v$ is adjacent to at least one of the $t \geq 3$ vertices of $I$ it dominates, it cannot belong to $I$ as $I$ is independent.  Therefore, the vertices of $I$ dominated by $v \notin I$ are adjacent to $v$.  In particular, at least $t$ vertices of $I$, all pairwise non-adjacent, are neighbors of $v$ and this induces $K_{1, t}$ in $G$.
\end{proof}

The preceding lemma shows that for any independent set $I$ in a $K_{1, t}$-free graph $G$, $|I| \leq (t - 1)\gamma(G)$.  Given that GREEDY outputs an independent set we obtain the following result which is of interest to us.

\begin{theorem}\label{thm : UB Claw-free}
$\rho($GREEDY, $K_{1, t}$-FREE$) = t - 1$.
\end{theorem}
\begin{proof}
The upper bound is a consequence of Proposition $\ref{lemma : I.S. Claw-Free}$ and the remarks that follow.  The lower bound follows from Theorem $\ref{thm : LB T-Claw-Free Graphs}$.
\end{proof}

\section{Noncompetitive Graph Classes}
\label{sec:noncompetitive}

Recall that the setting defined in \cite{OnlineDominatingSet1} is nearly identical to ours except that an algorithm knows the input size $n$ beforehand and the induced subgraph on the revealed vertices is not necessarily connected.  Within this setting the authors establish a lower bound of $\Omega(\sqrt{n})$ for arbitrary graphs.  Their proof can be augmented to show a lower bound of $\Omega(\sqrt{n})$ in our model, which is tight by our upper bound of $O(\sqrt{\Delta})$ on degree at most $\Delta$ graphs (applied to $\Delta = n-1$). Instead, we strengthen such a result in several ways by showing that the lower bound of $\Omega(\sqrt{n})$ applies to several restricted classes such as threshold graphs\footnote{With the caveat that, for threshold graphs, we instead consider the performance ratio as a function of input size.}, planar bipartite graphs, and series-parallel graphs.    

\subsection{Threshold Graphs}
\label{ssec:threshold}

The graph join operation applied to two graphs $G_1$ and $G_2$ takes the disjoint union of the two graphs and adds all possible edges between the two graphs to the result (in addition to retaining the edges of $G_1$ and $G_2$). The class of threshold graphs can be described recursively as follows:

\begin{enumerate}
    \item $K_1$ (i.e. a single isolated vertex) is a threshold graph.
    
    \item If $G$ is a threshold graph then the disjoint union $G \cup K_1$ is a threshold graph.
    
    \item If $G$ is a threshold graph then the graph join $G \oplus K_1$ is a threshold graph. 
    
\end{enumerate}

It is not hard to see that any connected threshold graph has a dominating set of size $1$.  Since our setting only allows for connected graphs we instead measure $ALG$ as a function of input size $n$ since $OPT \leq 1$ on every input.  In particular, we show that for any algorithm there is an infinite family of threshold graphs for which this algorithm selects $\Omega(\sqrt{n})$ vertices (where the input has $n$ vertices).  Although $OPT$ does not tend towards infinity, we consider this to be an asymptotic lower bound, but with input size $n$ tending to infinity.  In a sense, this is a stronger lower bound since the algorithm is guaranteed an input graph with a single dominating vertex, yet it still selects more than $\Omega(\sqrt{n})$ vertices in the input.

\begin{observation}\label{obs: StarThreshold}
The star on $n \geq 1$ vertices, that is, $K_{1, n - 1}$, is a threshold graph.
\end{observation}

Now we describe a slightly more complicated graph belonging to the class of threshold graphs.  Let $k \geq 2$ and $j_1, ..., j_{k}$ be non-negative integers.  Let $n = k + 1 + \sum\limits_{i = 1}^{k} j_i$ and consider the following graph $G$ on $n$ vertices; $V(G) = \{ u \} \cup C_k \cup I$, where $C_k = \{ v_1, ..., v_k \}$ and $I = I_{j_1} \cup I_{j_2} \cup ... \cup I_{j_k}$, with each $I_{j_i}$ having exactly $j_i$ vertices (each $I_{j_i}$ is possibly empty).  The set $\{ u \} \cup C_k$ is a clique on $k + 1$ vertices, and for each $i$, $I_{j_i}$ is an independent set where each $v \in I_{j_i}$ is adjacent only to vertices $v_i, ..., v_k$.  

\begin{lemma}\label{obs : Near-Clique With Leaves}
The graph described above is a threshold graph.  
\end{lemma}
\begin{proof}
We describe a construction using the recursive definition given above.  Initially, start with the single isolated vertex $u$.  For each $1 \leq i \leq k$, take the resulting graph from the previous step, disjoint union said graph with an independent set $I_{j_i}$ (i.e. repeatedly perform $j_i$ disjoint unions of with a single vertex) and then join the vertex $v_i$.  That is, let $G_0 = (\{ u \}, \emptyset)$ and for $1 \leq i \leq k$, $G_i = (G_{i - 1} \cup I_{j_i}) \oplus v_i$.
\end{proof}

We are now ready to prove a strong lower bound for any online algorithm.  Although we do not mention this explicitly in the proof, the adversarial inputs given are either $K_{1, k - 1}$ for some $k \geq 3$ or one that can be obtained by appropriately applying the recursive construction in Lemma~\ref{obs : Near-Clique With Leaves}.

\begin{theorem}\label{thm : LB Threshold Graphs}For infinitely many values of $n$ there is a threshold graph $G_n$ such that $$ ALG(G_n) = \Omega(\sqrt{n}).$$
\end{theorem}
\begin{proof}
Let $k \geq 3$ be an integer and reveal $v_1$ with $k - 1$ children.  If $ALG$ does not select $v_1$ then the input terminates as a star on $k$ vertices (i.e. the $k - 1$ neighbours of $v_1$ are revealed with no additional neighbours).  $ALG$ is forced to select the $k - 1$ neighbours of $v_1$ ($OPT$ selects only $v_1$).  In this case, the statement follows since $ALG = k - 1 \geq \sqrt{k} = \sqrt{n}$.  

Suppose that $ALG$ selects $v_1$ and let $c_i, 1 \leq i \leq k - 1$ be the children of $v_1$.  Reveal $c_1$ as adjacent to each child of $v_1$ and with an additional $k$ children. If $ALG$ does not select $c_1$ then the children of $c_1$ are revealed as leaves whereas the rest of the input is revealed to be a clique.  That is, $N[v_1]$ is a clique and only $c_1$ has children.  In this case, $ALG$ must select the $k$ children of $v_2$ yielding an output of $k + 1$ (see Figure \ref{fig:threshold-second-rejected} for an example).  Therefore, in this case the statement follows since $ALG \geq k + 1 = \frac{n}{2} + 1 \geq \sqrt{n}$.

Suppose that $ALG$ selects $c_1$, the input then continues in the following way; For each $2 \leq j \leq k - 1$, (as long as $ALG$ is accepting $c_j$) we reveal $c_j$ as adjacent to every visible vertex and with an additional $k$ children.  That is, $c_j$ is adjacent to each child $c_i, i \neq j$ of $v_1$ and the grandchildren of $v_1$ (i.e. the children of all the $c_i$ with $1 \leq i \leq j$) so that $c_j$ is a single dominating vertex of this prefix.

$\textbf{Case 1 :}$ If there is some $2 \leq j \leq k - 1$ such that $ALG$ does not select $c_j$ then the $k$ children of $c_j$ are revealed as leaves and $N[v_1]$ is revealed as a clique (see Figure \ref{fig:threshold-third-rejected} for an example).  The input has $n = 1 + (k - 1) + \sum\limits_{i = 2}^{j}ik = k + (j - 1)k = jk$ vertices.  At this point, $ALG$ has selected $\{ v_1, c_1, ..., c_{j - 1} \}$ and is now forced to select the $k$ children of $c_j$ ($OPT$ selects only $c_j$) for an output of at least $j + k \geq 2 + k \geq \sqrt{n}$ since $j < k$ and $n = jk$.  

$\textbf{Case 2 :}$ If $ALG$ selects each $c_i, 1 \leq i \leq k - 1$ then the input is terminated with $n = 1 + (k - 1) + \sum\limits_{i = 2}^{j}ik = k + (k - 1)k = k^{2}$ vertices after revealing $c_{k - 1}$.  $ALG$ has already selected $\{ v_1, c_1, ..., c_{k - 1} \}$ and therefore $ALG \geq k = \sqrt{n}$. 
\end{proof}

\begin{figure}[h]
\centering
\includegraphics[scale=0.6]{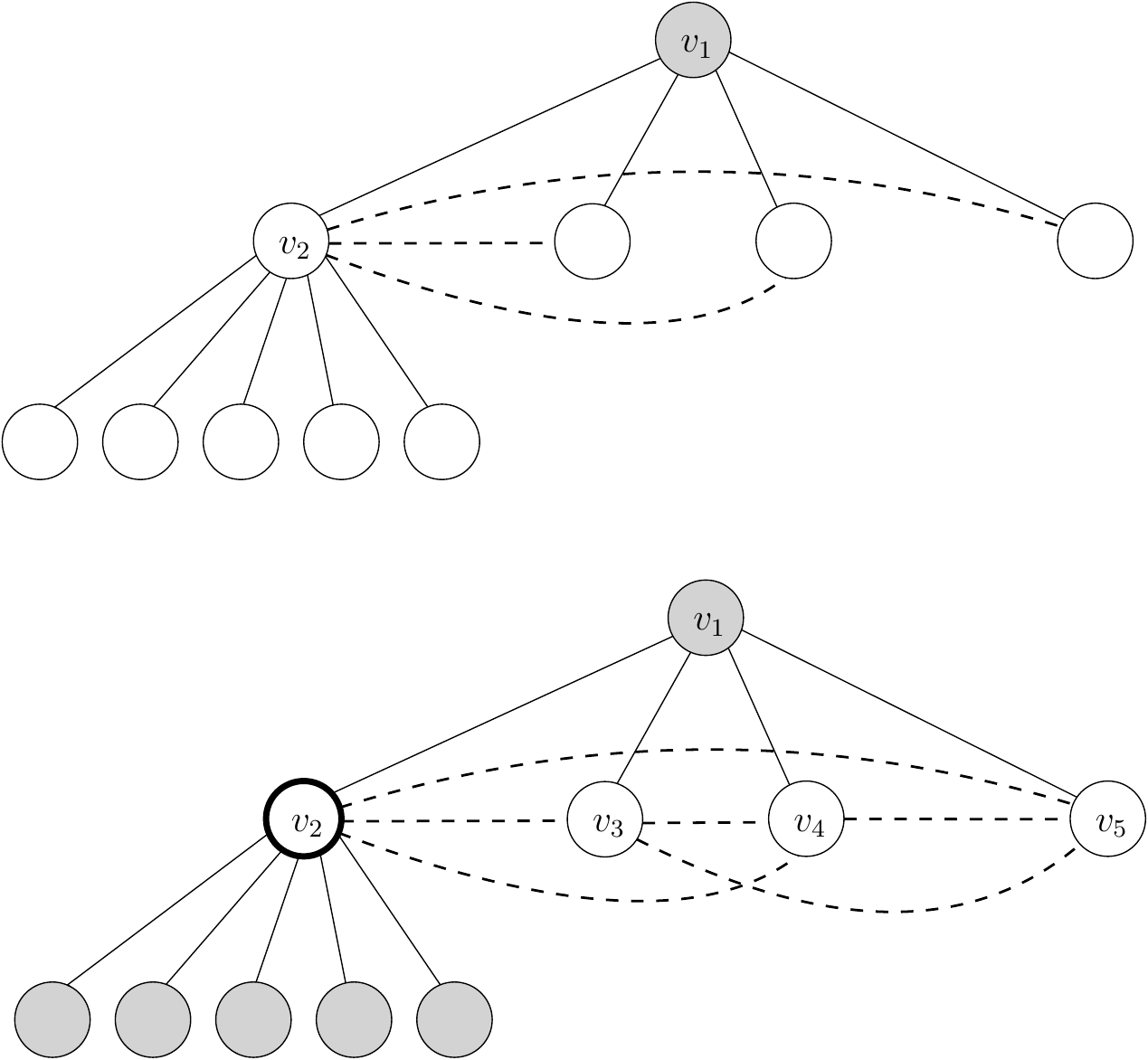}
\caption{An instance described in Theorem~\ref{thm : LB Threshold Graphs}  with $k = 5$ where $ALG$ does not select $v_2$.  The top depicts the graph at the moment $v_2$ was revealed and the bottom depicts the completely revealed graph.}\label{fig:threshold-second-rejected}
\end{figure}

\begin{figure}[h]
\centering
\includegraphics[scale=0.6]{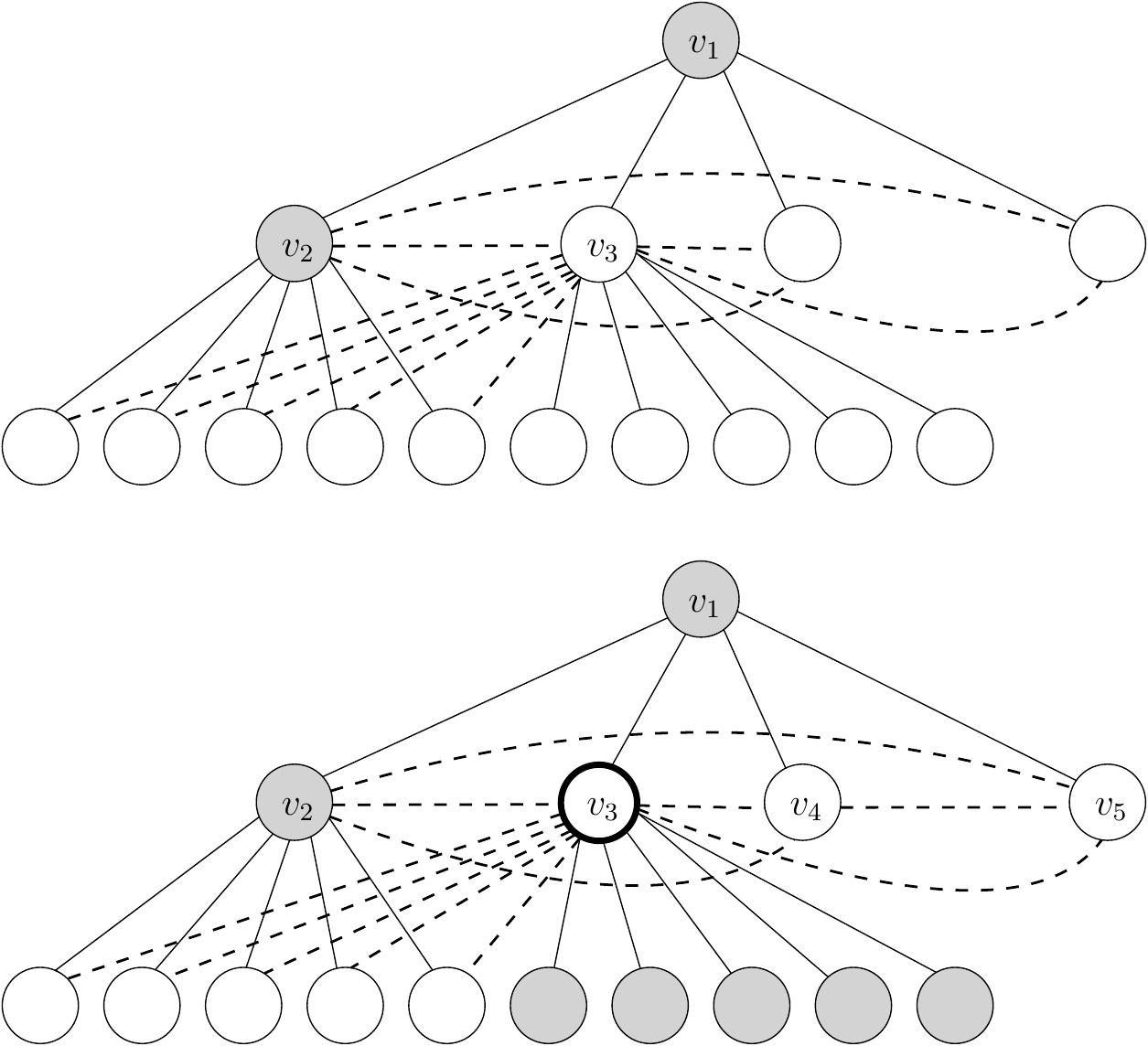}
\caption{An instance described in Theorem~\ref{thm : LB Threshold Graphs}  with $k = 5$ where $ALG$ does not select $v_3$.  The top depicts the graph at the moment $v_3$ was revealed and the bottom depicts the completely revealed graph.}\label{fig:threshold-third-rejected}
\end{figure}

\subsection{Planar Bipartite Graphs}
\label{ssec:planar}

Below is a lower bound of $\Omega(\sqrt{n})$ for planar bipartite graphs.  We should mention that is strikingly similar to the lower bound on general graphs given in \cite{OnlineDominatingSet1}.  We provide a simple augmentation of their lower bound so that it not only consists of inputs that are revealed according to our model but inputs that are also planar bipartite graphs.

\begin{theorem}\label{thm : LB Bipartite Graphs}
$ \rho(ALG, $PLANAR BIPARTITE$) = \Omega(\sqrt{n}).$
\end{theorem}
\begin{proof}
Let $k \geq 2$ and consider a path on $k$ vertices with the vertices ordered $v_i, 1 \leq i \leq k$.  Each vertex along the path is adjacent to $k$ neighbors appearing as leaves.  Every odd labeled vertex is adjacent to a common vertex $o$ and every even labeled vertex is adjacent to a common vertex $e$ where both $e, o$ do not lie on the path (and have not yet been revealed).  The ordering of the path is the order in which these vertices were revealed to $ALG$ (see Figure~\ref{fig:planar-bipartite-1-instance}). Of the $k$ vertices along the path we suppose that $ALG$ selects $k - i$  where $0 \leq i \leq k$.  For each of the $k - (k - i) = i$ vertices not selected by $ALG$, the $k$ leaves adjacent are revealed to remain leaves and $ALG$ must select them.  For each of the $k - i$ vertices selected by $ALG$, the leaves adjacent to said vertices are revealed as adjacent to $e$ if their neighbour had an odd label and $o$ if their neighbour had an even label.  Thus, $ALG$ selects at least $(k - i) + (i)k = k + i(k - 1)$ vertices whereas $OPT$ need only select $o, e$ and the $i$ vertices not selected by $ALG$.  

Thus, we have that $\frac{ALG}{OPT} \geq \frac{k + i(k-1)}{i + 2}$.  Noting that $k - 1 \geq \frac{k}{2}$ since $k \geq 2$ we obtain that 
$$ i(k - 1) \geq i(k /2) \iff 
k + i(k - 1) \geq i(k / 2) + k = \frac{k}{2}(i + 2) \iff
\frac{k + i(k-1)}{i + 2} \geq \frac{k}{2}
.$$

Since the input consists of $n = k^{2} + k + 2$ vertices the result would then follow.  To finish we provide a justification that the input is planar and bipartite.  To see that it is bipartite let one part $X$ consist of the vertices along the path with odd labels and the neighbors of the vertices with even labels (this includes $e$).  The other part $Y$ consists of the vertices along the path with even labels and the neighbors of the vertices with odd labels (this includes $o$).  To see that it is planar, consider a drawing with the vertices along the path drawn in a line from left to right, $e$ placed above this path and $o$ placed below.  For any odd labeled vertex $v_i$, the $k$ neighbors of $v_i$ that do not lie on the path (and are different from $o$) are placed immediately above $v_i$ but below $e$ (i.e. $v_i$ along with said neighbors are depicted as a star on $k + 1$ vertices with $v_i$ as the center).  Similarly, for any even labeled vertex $v_i$, the $k$ neighbors of $v_i$ that do not lie on the path and are different from $e$ are placed immediately below $v_i$ and above $o$.
\end{proof}

\begin{figure}[h]
\centering
\includegraphics[scale=0.5]{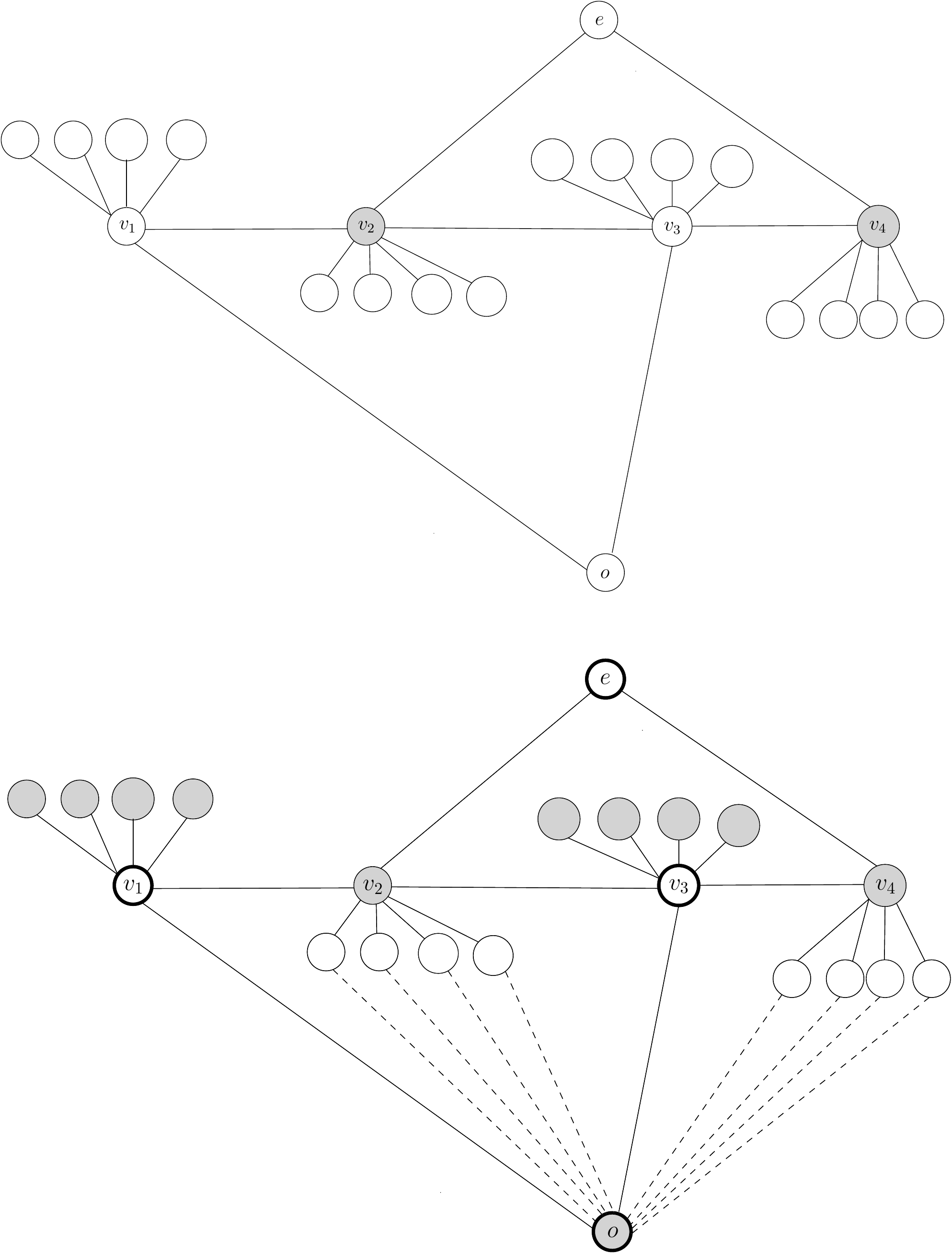}
\caption{An instance described in Theorem~\ref{thm : LB Bipartite Graphs}  with $k = 4$.  The top depicts the prefix where only vertices along the path have been revealed.  Assuming the vertices on the path that $ALG$ selects are $v_2$ and $v_4$, the bottom depicts the completely revealed graph.}\label{fig:planar-bipartite-1-instance}
\end{figure}

We remark that in Theorem~\ref{thm : LB Bipartite Graphs} there are cases when $OPT \leq \alpha$ for some constant $\alpha \geq 2$.  For example, when $ALG$ selects all $k$ vertices along the path $OPT$ selects only $\{ e, o \}$.  In this case, we extend the input by revealing $o$ with an additional neighbor $u_1$, and repeat a similar trap with $u_1$ as the first vertex along the path.

\subsection{Series-Parallel Graphs}
\label{ssec:series-parallel}

In light of our $2$-competitive algorithm for trees, it is natural to suppose that some class of graphs generalizing trees might admit competitive algorithms, that is, algorithms with bounded competitive ratio. One such generalization is graphs of bounded treewidth. Trees have treewidth $1$, so the next step is to consider graphs of treewidth $2$. Unfortunately, in this section we show that by increasing treewidth parameter from $1$ to $2$, the online dominating set problem becomes extremely hard for online algorithms. More specifically, we show that series-parallel graphs do not admit online algorithms with competitive ratio better than $\Omega(\sqrt{n})$. We remark that series-parallel graphs have treewidth at most $2$.

We begin by recalling the definition of a series-parallel graph. It is defined with the help of the notion of a two-terminal graph $(G,s,t)$, which is a graph $G$ with two distinguished vertices $s$, called a source, and $t$, called a sink. For a pair of two-terminal graphs $(G_1, s_1, t_1)$ and $(G_2, s_2, t_2)$, there are two composition operations:
\begin{itemize}
    \item \emph{Parallel composition}: take a disjoint union of $G_1$ with $G_2$ and merge $s_1$ with $s_2$ to get the new source, as well as $t_1$ with $t_2$ to get the new sink.
    \item \emph{Series composition}: take a disjoint union of $G_1$ with $G_2$ and merge $t_1$ with $s_2$, which now becomes an inner vertex of the resulting two-terminal graph; $s_1$ becomes the new source and $t_2$ becomes the new sink.
\end{itemize}
A two-terminal series-parallel graph is a two-terminal graph that can be obtained by starting with several copies of the $K_2$ graph and applying a sequence of parallel and series compositions. Lastly, a graph is called series-parallel if it is a two-terminal series-parallel graph for some choice of source and sink vertices. Observe that intermediate graphs resulting in the construction of a series-parallel graph may have multiple parallel edges, so they are multigraphs. This is permitted, as long as the resulting overall graph is a simple undirected graph at the end.

Now, we are ready to prove the main result of this section.

\begin{theorem}\label{thm : LB SP Graphs}
$ \rho(ALG,$SERIES-PARALLEL$) = \Omega(\sqrt{n}).$
\end{theorem}
\begin{proof}
Let $k \ge 2$ be an integer. The adversary reveals $s$ with $k$ neighbors $c_1, \ldots, c_k$. Then $c_1, \ldots, c_k$ are revealed in this order with $k$  new neighbors each. Let neighbors of $c_i$ be $d_{i1}, \ldots, d_{ik}$. Let $S \subseteq \{c_1, \ldots, c_k\}$ be those vertices selected by $ALG$. For those $i \notin S$ we reveal their new neighbors in order $d_{i1}, \ldots, d_{ik}$. Each such $d_{ij}$ is revealed with a single new neighbor $f_{ij}$. For $i \in S$ we reveal their new neighbors in order $d_{i1}, \ldots, d_{ik}$. Each such $d_{ij}$ is revealed with a new neighbor $t$ that is common to all these vertices. Then $f_{ij}$ are revealed in arbitrary order with $t$ as a new neighbor. Lastly $t$ is revealed without any new neighbors.

Let $p = |S|$. Observe that in addition to these $p$ vertices $ALG$ must select at least one vertex from each of $\{d_{ij}, f_{ij}\}$ pairs for those $i \notin S$; otherwise, vertex $d_{ij}$ would be undominated. Thus, $ALG \ge p + k(k-p)$. Also, observe that $\{s, t\} \cup \{c_i \mid i \notin S\}$ is a dominating set, so $OPT \le k-p+2$. The bound on the competitive ratio is
\[\frac{ALG}{OPT} \ge \frac{p + k(k-p)}{k-p+2} = k - \frac{2k-p}{k-p+2} \ge \frac{k}{2},\]
where the last inequality is obtained as follows. For $k \ge 2$ we have $k^2 -k p \ge 2k - 2p$, which implies $k^2-kp + 2k \ge 4k-2p$. This in turn implies that $k(k-p+2) \ge 2(2k-p)$, hence $(2k-p)/(k-p+2) \le k/2$. The quantitative part of the statement of this theorem follows from the fact that the total number of vertices is at most $2+k+k^2+k(k-p) = \Theta(k^2)$. 

Lastly, we note that the adversarial graph thus constructed is, indeed, series-parallel. For each $i \notin S$ and $j \in \{1, \ldots, k\}$ the path $c_i \rightarrow d_{ij} \rightarrow f_{ij} \rightarrow t$ is a series-composition of $3$ copies of $K_2$. These paths can be merged by a parallel composition to obtain the subgraph induced on $\{c_i, t\} \cup \{d_{ij}, f_{ij} \mid j \in \{1, \ldots, k\}\}$ for each $i \notin S$. Each of these subgraphs is composed at $c_i$ with another copy of $K_2$ with the new vertex playing the role of $s$. Similar argument holds to show that the subgraph induced on $\{s, c_i, t\} \cup \{d_{ij} \mid j \in \{1, \ldots, k\}\}$ for $i \in S$ is a two-terminal series-parallel graph. Lastly, all these subgraphs are merged by a sequence of parallel compositions at $s$ and $t$.
\end{proof}

\begin{figure}[h]
\centering
\includegraphics[scale=0.5]{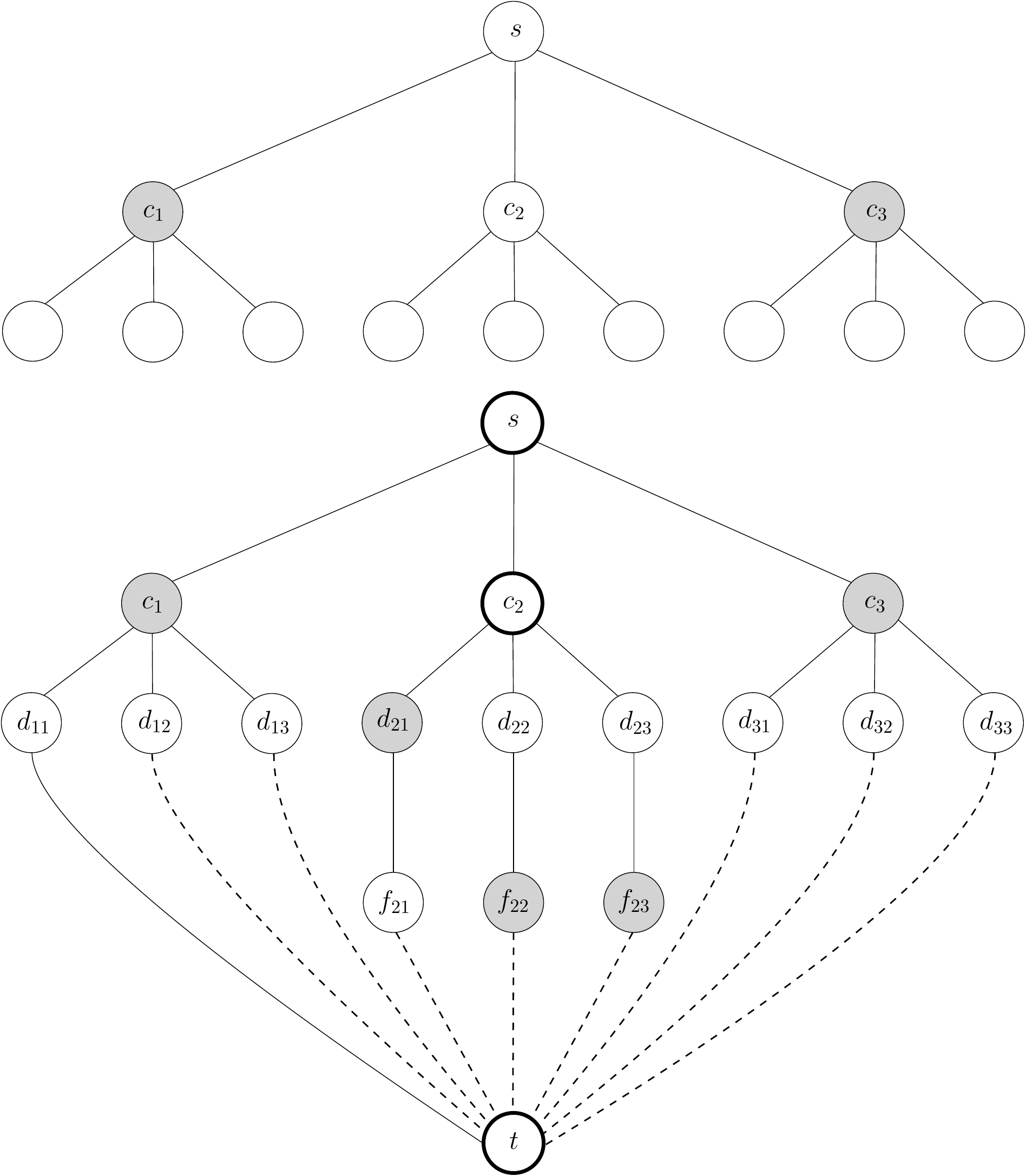}
\caption{An instance described in Theorem~\ref{thm : LB SP Graphs}  with $k = 3$.  The top depicts a prefix where $ALG$ selects $S = \{ c_1, c_3 \}$ whereas the bottom depicts the completely revealed graph.}\label{fig:series-parallel-lower-bound}
\end{figure}

\section{Conclusions}

In this paper we studied the minimum dominating set problem in an online setting where a vertex is revealed alongside all its neighbors. We also contrasted our results with those obtained by Boyar et al.~\cite{OnlineDominatingSet3} and Kobayashi~\cite{OnlineDominatingSet4} in a related vertex-arrival model. Dominating set is a difficult problem both offline and online. In our setting, the best achievable competitive ratio on general graphs is $O(\sqrt{n})$. This observation prompted us to study this problem with respect to more restrictive graph classes. Trees provide a natural graph class that usually allows for non-trivial competitive ratios. Indeed, we showed that in our model trees admit $2$-competitive algorithms. There are several ways to try to extend this result to larger graph classes. We considered cactus graphs and showed that the optimal competitive ratio is $2.5$ on them. Another way of generalizing trees is to consider graphs of higher treewidth. Unfortunately, once treewidth goes up to $2$, competitive ratio jumps to $\Omega(\sqrt{n})$ (which is trivial in our setting due to $O(\sqrt{n})$ upper bound), as witnessed by series-parallel graphs. We also established non-trivial upper bounds on graphs of bounded degree, as well as graphs with bounded claws. When one moves to planar (even bipartite planar) graphs and threshold graphs, the competitive ratio jumps to $\Omega(\sqrt{n})$ again.

The above can be viewed as a larger program of developing a deeper understanding of the dominating set problem in an online setting. What are the main structural obstacles in graphs that prohibit online algorithms with small competitive ratios? Can one discover a family of graphs parameterized by some parameter $t$, which include cactus graphs, claw-free graphs, and bounded-degree graphs, such that the competitive ratio scales gracefully with $t$? Lastly, as another research direction, we mention that we have only considered the deterministic setting, so it would be of interest to extend our results to the randomized setting, as well as the setting of online algorithms with advice.

%
%
%
\bibliographystyle{splncs04}
\bibliography{refs}

\begin{thebibliography}{10}
\providecommand{\url}[1]{\texttt{#1}}
\providecommand{\urlprefix}{URL }
\providecommand{\doi}[1]{https://doi.org/#1}

\bibitem{Berge62}
Berge, C.: The Theory of Graphs and Its Applications. Methuen (1962)

\bibitem{BorodinE1998}
Borodin, A., El{-}Yaniv, R.: Online Computation and Competitive Analysis.
  Cambridge University Press (1998)

\bibitem{OnlineDominatingSet3}
Boyar, J., Eidenbenz, S.J., Favrholdt, L.M., Kotrbc{\'{\i}}k, M., Larsen, K.S.:
  Online dominating set. Algorithmica  \textbf{81}(5),  1938--1964 (2019)

\bibitem{FPTClawFree}
Cygan, M., Philip, G., Pilipczuk, M., Pilipczuk, M., Wojtaszczyk, J.O.:
  Dominating set is fixed parameter tractable in claw-free graphs. Theoretical
  Computer Science  \textbf{412}(50),  6982--7000 (2011)

\bibitem{DasB97}
Das, B., Bharghavan, V.: Routing in ad-hoc networks using minimum connected
  dominating sets. In: 1997 {IEEE} International Conference on Communications:
  Towards the Knowledge Millennium, {ICC} 1997, Montr{\'{e}}al, Qu{\'{e}}bec,
  Canada, June 8-12, 1997. pp. 376--380. {IEEE} (1997)

\bibitem{VertexAddition}
Harutyunyan, H.A.: An efficient vertex addition method for broadcast networks.
  Internet Math.  \textbf{5}(3),  211--225 (2008)

\bibitem{BroadcastingDomination}
Harutyunyan, H.A., Liestman, A.L.: Upper bounds on the broadcast function using
  minimum dominating sets. Discret. Math.  \textbf{312}(20),  2992--2996 (2012)

\bibitem{HaynesHS98}
Haynes, T., Hedetniemi, S., Slater, P.: Fundamentals of Domination in Graphs.
  Marcel Dekker, New York (1998)

\bibitem{HenningA13}
Henning, M., Yeo, A.: Total Domination in Graphs. Springer-Verlag New York
  (2013)

\bibitem{OnlineDominatingSet1}
King, G., Tzeng, W.: On-line algorithms for the dominating set problem. Inf.
  Process. Lett.  \textbf{61}(1),  11--14 (1997)

\bibitem{OnlineDominatingSet4}
Kobayashi, K.M.: Improved bounds for online dominating sets of trees. In:
  Okamoto, Y., Tokuyama, T. (eds.) 28th International Symposium on Algorithms
  and Computation, {ISAAC} 2017, December 9-12, 2017, Phuket, Thailand. LIPIcs,
  vol.~92, pp. 52:1--52:13. Schloss Dagstuhl - Leibniz-Zentrum f{\"{u}}r
  Informatik (2017)

\bibitem{Komm16}
Komm, D.: An Introduction to Online Computation - Determinism, Randomization,
  Advice. Texts in Theoretical Computer Science. An {EATCS} Series, Springer
  (2016)

\bibitem{Konig50}
K\"{o}nig, D.: Theorie der Endlichen und Unendlichen Graphen. Chelsea, New York
  (1950)

\bibitem{Ore62}
Ore, O.: Theory of Graphs. American Mathematical Society (1962)

\bibitem{WangDC16}
Wang, F., Du, D., Cheng, X.: Connected dominating set. In: Encyclopedia of
  Algorithms, pp. 425--430 (2016)

\end{thebibliography}

\end{document}